\documentclass[journal]{IEEEtran}
\usepackage[cmex10]{amsmath}
\usepackage{eqparbox}
\usepackage[usenames, dvipsnames]{color}
\usepackage{mathtools,amssymb,bm,mathabx}
\usepackage{amstext}
\usepackage{amssymb}
\usepackage{graphicx}
\usepackage{color}
\usepackage{booktabs}
\usepackage{longtable}
\usepackage{multicol}
\usepackage{multirow}
\usepackage{amsfonts}
\usepackage{dsfont}
\usepackage{array}
\usepackage[ruled]{algorithm}
\usepackage{algorithmic}
\usepackage{cases}
\usepackage{pgfplots}
\usepackage{accents}
\usepackage[center]{caption}
\usepackage[footnotesize]{subfigure}
\usepackage{amsthm,xparse}
\DeclareCaptionLabelSeparator{periodspace}{.\quad}
\usepackage{float}
\usepackage{cite}
\usepackage [autostyle, english = american]{csquotes}
\usepackage{hyperref}

\theoremstyle{remark}

\newcommand\ASTART{\bigskip\noindent\begin{minipage}[b]{0.5\linewidth}}
	
	\newcommand\AENDSKIP{\end{minipage}\bigskip}
\newcommand\AEND{\end{minipage}}
\ifCLASSOPTIONcaptionsoff
\usepackage[nomarkers]{endfloat}
\let\MYoriglatexcaption\caption
\renewcommand{\caption}[2][\relax]{\MYoriglatexcaption[#2]{#2}}
\fi
\theoremstyle{plain}
\newtheorem{thm}{\textbf{Theorem}}
\newtheorem{lem}{\textbf{Lemma}}

\theoremstyle{definition}
\newtheorem{defn}{\textbf{Definition}}

\theoremstyle{remark}
\newtheorem{rem}{\bf Remark}

\newcommand*{\rom}[1]{\expandafter\@slowromancap\romannumeral #1@}
\def\change{black}

\newcommand{\RN}[1]{%
\textup{\uppercase\expandafter{\romannumeral#1}}%
}

\usepackage{standalone}
\graphicspath{{./Figures/}} 
\begin{document}
%
\title{Multi-weight Nuclear Norm Minimization for Low-rank Matrix Recovery in Presence of Subspace Prior Information}
\author{Hamideh Sadat~Fazael~Ardakani, Sajad~Daei, Farzan~Haddadi%
	\thanks{H S. Fazael Ardakani is with the School of Electrical
		and Computer Engineering, University of Tehran, Iran. S. Daei is with the Communications Department of Eurecom, Sophia Antipolis, France. F. Haddadi is with the School of Electrical Engineering, Iran University of Science \& Technology, Iran.}}
\maketitle

	\begin{abstract}
		Weighted nuclear norm minimization has been recently recognized as a technique for reconstruction of a low-rank
		matrix from compressively sampled measurements when some prior information about the column and row subspaces of the matrix is available. We derive the conditions and the associated recovery guarantees of weighted nuclear norm minimization when multiple weights are allowed. This setup could be used when one has access to prior subspaces forming multiple angles with the column and row subspaces of the ground-truth matrix. While existing works in this field use a single weight to penalize all the angles, we propose a multi-weight problem which is designed to penalize each angle independently using a distinct weight. Specifically, we prove that our proposed multi-weight problem is stable and robust under weaker conditions for the measurement operator than the analogous conditions for single-weight scenario and standard nuclear norm minimization. Moreover, it provides better reconstruction error than the state of the art methods. We illustrate our results with extensive numerical experiments that demonstrate the advantages of allowing multiple weights in the recovery procedure. Our work has beneficial implications for channel estimation in multiple-input multiple output (MIMO) wireless communications based on Frequency Division Duplexing (FDD). The existing methods for channel estimation in this application require a huge number of pilot (training) signals to estimate the downlink channel which greatly wastes the spectrum resources in massive MIMO systems. We provide a dynamic channel estimation scenario for FDD massive MIMO systems and show how our method could be applied to enhance the spectral efficiency.
	\end{abstract}

%
\IEEEpeerreviewmaketitle

\section{Introduction}\label{introduction}
In many applications such as channel estimation in wireless communication \cite{FDD}, MRI \cite{haldar2010spatiotemporal},\cite{zhao2010low}, quantum state tomography \cite{gross2010quantum}, collaborative filtering \cite{srebro2010collaborative}, Netflix problem \cite{bennett2007netflix} and exploration seismology \cite{aravkin2014fast}, we are interested in recovering a low-rank matrix $\bm{X} \in \mathbb{R}^{n\times n}$ with rank $r$ from linear noisy measurements $ \bm{y} = \mathcal{A}(\bm{X}+\bm{E})\in \mathbb{R}^{p}$ by solving the following problem:
\begin{align}\label{minrank}
	&\min_{\bm{Z} \in \mathbb{R}^{n\times n}}~{\rm{rank}}(\bm{Z}) \nonumber \\
	&\mathrm{s.t.}~\|\bm{y} - \mathcal{A}(\bm{Z})\|_2\le e ,
\end{align}
where  $ \mathcal{A}:\mathbb{R}^{n\times n} \rightarrow \mathbb{R}^{p} $ is the linear measurement operator \footnote{All conclusions in this work are reasonable for non-square matrices, though without loss of generality we consider square matrices.}, $\bm{E}$ is the noise matrix ,  $\bm{Z}$ represents the matrix variable of the optimization problem and $e$ is an upper-bound for $\|\mathcal{A}(\bm{E})\|_2$.
The latter problem is NP-hard, so the common approach is to solve the surrogate convex problem 
\begin{align}\label{minnuclear}
	&\min_{\bm{Z} \in \mathbb{R}^{n\times n}}~\|\bm{Z}\|_{*} \nonumber \\
	&\mathrm{s.t.}~\|\bm{y} - \mathcal{A}(\bm{Z})\|_2\le e,
\end{align}
where $ \|\cdot\|_{*} $ is the nuclear norm \cite{recht2010guaranteed}. It was shown in \cite{recht2010guaranteed} that if $\mathcal{A}$ satisfies the rank restricted isometry property (R-RIP), then the problem \eqref{minnuclear} can (approximately) recover $ \bm{X} $. In many applications, some prior information about the ground-truth subspaces (i.e. the column and row subspaces of the ground-truth matrix $\bm{X}$) is available. In Netflix problem, prior evaluations of the movies might be available. In sensor network localization \cite{so2007theory}, previous positions might be available. We consider this prior information as two $r'$-dimensional subspaces $\widetilde{\bm{\mathcal{U}}}_{r^{\prime}}$ and $ \widetilde{\bm{\mathcal{V}}}_{r^{\prime}} $ forming angles with column and row spaces of the ground-truth  matrix $\bm{X}$, respectively. To incorporate this prior information into the recovery procedure,
we propose the following problem for low-rank matrix recovery:
\begin{align}\label{ourproblem}
	&\min_{\bm{Z} \in \mathbb{R}^{n\times n}} ~\|\bm{Q}_{\widetilde{\bm{\mathcal{U}}}_{r^{\prime}}} \bm{Z} \bm{Q}_{\widetilde{\bm{\mathcal{V}}}_{r^{\prime}}}\|_{*} \nonumber \\
	&\mathrm{s.t.}~\|\bm{y} - \mathcal{A}(\bm{Z})\|_2\le e,
\end{align}
where

\begin{align}\label{QUQV}
	&\bm{Q}_{\widetilde{\bm{\mathcal{U}}}_{r^{\prime}}} := \widetilde{\bm{U}}_{r^{\prime}}\bm{\Lambda}\widetilde{\bm{U}}_{r^{\prime}}^{\rm{H}} +\bm{P}_{\widetilde{\bm{\mathcal{U}}}^{\perp}_{r^{\prime}}}	 \nonumber \\
	&\bm{Q}_{\widetilde{\bm{\mathcal{V}}}_{r^{\prime}}} := \widetilde{\bm{V}}_{r^{\prime}}\bm{\Gamma}\widetilde{\bm{V}}_{r^{\prime}}^{\rm{H}} + \bm{P}_{\widetilde{\bm{\mathcal{V}}}^{\perp}_{r^{\prime}}},
\end{align}
and $ \bm{\Lambda} $ and $ \bm{\Gamma} $ are diagonal matrices with its entries in the interval $[0,1]$, $\widetilde{\bm{U}}_{r^{\prime}}\in\mathbb{R}^{n\times r'}$ and $\widetilde{\bm{V}}_{r^{\prime}}\in\mathbb{R}^{n\times r'}$ are some bases of the subspaces  ${\widetilde{\bm{\mathcal{U}}}_{r^{\prime}}}$ and ${\widetilde{\bm{\mathcal{V}}}_{r^{\prime}}}$, respectively. $\bm{P}_{\widetilde{\bm{\mathcal{U}}}^{\perp}_{r^{\prime}}}$ and $\bm{P}_{\widetilde{\bm{\mathcal{V}}}^{\perp}_{r^{\prime}}}$ are orthogonal projection matrices onto the complement subspaces $\widetilde{\bm{\mathcal{U}}}^{\perp}_{r^{\prime}}$ and $\widetilde{\bm{\mathcal{V}}}^{\perp}_{r^{\prime}}$, respectively defined by $\bm{P}_{\widetilde{\bm{\mathcal{U}}}^{\perp}_{r^{\prime}}} := \bm{I} - \widetilde{\bm{U}}_{r^{\prime}} \widetilde{\bm{U}}_{r^{\prime}}^{H},$ and $\bm{P}_{\widetilde{\bm{\mathcal{V}}}^{\perp}_{r^{\prime}}} := \bm{I} - \widetilde{\bm{V}}_{r^{\prime}} \widetilde{\bm{V}}_{r^{\prime}}^{H}$.

The problem \eqref{ourproblem} reduces to \eqref{minnuclear} when  $ \bm{\Lambda} = \bm{\Gamma} = \bm{I}_{r^{\prime}} $. The values of $\bm{\Lambda}$ and $\bm{\Gamma}$ depend on the accuracy of prior information for each direction (e.g. each column of $\widetilde{\bm{U}}_{r^{\prime}}$ ) in the form of principal angles \footnote{ See Section \ref{section.Existing results} for a detailed definition of principal angles}. Whenever a principal angle increases, the accuracy of the corresponding direction decreases, and therefore the weight being assigned to that direction shall intuitively be large and near $ 1 $.

	This prior information is accessible in many applications \cite{srebro2010collaborative,FDD,aravkin2014fast,netflix,shen2016compressed}. For example, in wireless communication systems based on frequency division duplex (FDD), the base station (BS) equipped with multiple antennas transmits a few pilots (training signal) to the single-antenna users in the downlink and each user estimates its own channel in a coherence time-bandwidth block based on this observation and feedbacks the estimates to the BS \cite{chen2020massive}-\cite{gao2015structured}. 
	The number of required pilots grows linearly with the number of BS antennas. Hence, in massive multiple-input multiple output (MIMO) systems, the overhead incurred by pilot signaling imposes a serious concern and becomes highly challenging. The low-rank structure of the channel matrix between users and BS (which is the result of few scatterers in the communication path and is the case in millimeter wave systems \cite{chen2020massive,liang2019semi,shen2016compressed}) can help to reduce the number of required pilots via exploiting low-rank matrix recovery \eqref{minnuclear}. However, further reducing the pilot overhead (which could be also translated to further enhancing the spectral efficiency) is possible by leveraging additional information coming from previous coherent time-bandwidth blocks. Specifically, the associated Doppler frequency of channel specifies a maximum level of dissimilarity between the matrix channels at consecutive coherence blocks. In simple words, the angles between column/ row subspaces of the ground-truth channel matrix (say e.g. $\bm{\mathcal{U}}$/ $\bm{\mathcal{V}}$) and the column/row subspaces of channel matrix in a previous coherence block (represented by $\widetilde{\bm{\mathcal{U}}}$/$\widetilde{\bm{\mathcal{V}}}$) could be approximately estimated in advance. This extra information can help to further enhance the spectral efficiency and reduce the number of required pilots. Our work here provides a method to use this information by solving \eqref{ourproblem} and by designing optimal weights $\bm{\Lambda}$ and $\bm{\Gamma}$. We will disclose more details of this application in Section \ref{sec.fdd} and will explain how our method could be applied to further reduce the number of required pilots which in turn enhances the spectral efficiency.
%



\subsection{Contributions}
In this paper, we propose a general problem for low-rank matrix recovery with prior subspace information. For a fixed linear operator, we guarantee that our method outperforms the existing methods in \cite{eftekhari2018weighted} and \cite{recht2010guaranteed}, in terms of the estimation error, since we penalize the inaccuracy of each basis (direction) in the prior subspace, distinctly. We derive an RIP condition for the measurement operator in this multi-weight weighted matrix recovery that is weaker than its single-weight counterpart. Then, we obtain the optimal weights that make the condition as weak as possible.

\subsection{Related Works and Key Differences}\label{Related Works}
In this section, we summarize the existing approaches for recovering low-rank matrix from linear measurements. 
The authors in \cite{rao2015collaborative} provide a weighted version of trace-norm regularization that works better than the unweighted version: 
\begin{align}\label{tracenorm}
	\|\bm{X}\|_{\rm tr} := \| {\rm diag}(\sqrt{\bm{p}}) \bm{X} {\rm diag}(\sqrt{\bm{q}}) \|_{*},
\end{align}
where $p(i)$ and $q(j)$ are the probabilities of the $ i $-th row and $ j $-th column of the matrix being observed, respectively.

In \cite{angst2011generalized}, \cite{jain2013provable} and \cite{xu2013speedup}, prior information is used to penalize the directions in row and column spaces of $ \bm{X} $. 
In \cite{mohan2010reweighted}, the authors consider re-weighted trace norm minimization problem as an iterative heuristic and analyze its convergence.
In \cite{rao2015collaborative}, a generalized nuclear norm from \cite{srebro2010collaborative} is used to provide a scalable algorithm based on \cite{zhou2012kernelized} with structural prior information for matrix recovery.

Aravkin et al. in \cite{aravkin2014fast} were the first team that incorporated prior subspace information into low-rank matrix recovery using an iterative algorithm to solve the following problem:
\begin{align}\label{arminproblem}
	&\min_{\bm{Z} \in \mathbb{R}^{n\times n}} ~\|\bm{Q}_{\widetilde{\bm{\mathcal{U}}}_{r}} \bm{Z} \bm{Q}_{\widetilde{\bm{\mathcal{V}}}_{r}}\|_{*} \nonumber \\
	&\mathrm{s.t.}~\bm{y} = \mathcal{A}(\bm{Z}),
\end{align}
where

\begin{align}\label{QUQVone}
	&\bm{Q}_{\widetilde{\bm{\mathcal{U}}}_{r}} := \lambda\bm{P}_{\widetilde{\bm{\mathcal{U}}}_{r}} +\bm{P}_{\widetilde{\bm{\mathcal{U}}}^{\perp}_{r}}	 \nonumber \\
	&\bm{Q}_{\widetilde{\bm{\mathcal{V}}}_{r}} := \gamma \bm{P}_{\widetilde{\bm{\mathcal{V}}}_{r}} + \bm{P}_{\widetilde{\bm{\mathcal{V}}}^{\perp}_{r}},
\end{align}
and  $ \lambda $ and $ \gamma $ depend on the maximum principal angle.The intuition behind using $\bm{Q}_{\widetilde{\bm{\mathcal{U}}}_{r}}$ and $\bm{Q}_{\widetilde{\bm{\mathcal{V}}}_{r}}$ is forming an objective function that promotes both rank and additional prior subspace information.

Eftekhari et al. in \cite{eftekhari2018weighted} uses the problem \eqref{arminproblem} and proves that the isometry constant of the linear operator for robust matrix recovery is weaker in the presence of prior information.
In \cite{ardakani2019greedy}, a greedy method is provided to solve rank minimization problem according to \eqref{arminproblem}. The prior subspace information in \cite{ardakani2019greedy} might be close or far from the ground-truth subspaces in contrast to \cite{eftekhari2018weighted} and \cite{aravkin2014fast} where prior subspaces must be close to the ground-truth subspaces.

Another work with the same model as \eqref{arminproblem} is \cite{daei2018optimal}, which uses statistical dimension theory to obtain optimal weights that minimize the required number of measurements in contrast to other works that maximize the RIP bound. In a closely related field known as compressed sensing(CS)\cite{donoho2006compressed,candes2008restricted,daei2019error,daei2019living}, Needell et al. in \cite{needell2017weighted} provide recovery conditions for weighted $ \ell_1 $-minimization when multiple prior information about the support of a sparse signal is available. This prior information appears in the form of multiple sets where each contributes to the support with a certain accuracy and non-uniform weights are assigned to these sets. It is worth noting that the terms "non-uniform weights" refers to the multiple distinct weights; we will use both terms in this paper. In another CS-related work applied to downlink channel estimation in FDD massive MIMO, \cite{lu2019compressive} proposes a weighted $\ell_p$ minimization to estimate the sparse channel and chooses the weights based on the previous obtained channel support. However, besides different setup with the considered model of our work, the method of choosing the weights is not optimal. Our work in this paper is actually an extension of \cite{needell2017weighted} to the matrix recovery case. We use non-uniform weights to penalize different directions of the ground-truth matrix. We should point out that our used tools and analysis substantially differ from those in \cite{needell2017weighted}. In particular, while we obtain the optimal weights that maximize the RIP bound (alternatively make the RIP condition as weak as possible), the weights in \cite{needell2017weighted} are chosen in a heuristic way.   
\subsection{Outline and Notations}
The paper is organized as follows:In Section \ref{section.Existing results}, we review the results on weighted nuclear norm minimization \eqref{ourproblem} with a single weight. In Section \ref{ourwork}, we present a generalized and improved theory of non-uniform weighted nuclear norm minimization. {\color{\change} As one typical application of the proposed method, we explain the logic linking between low-rank matrix recovery and channel estimation in FDD massive MIMO in Section \ref{sec.fdd}}. Numerical results is provided in Section \ref{section.simulation}. Finally the paper is concluded in Section \ref{conclusion}.

Throughout the paper, scalars are indicated by lowercase letters, vectors by lowercase boldface letters, and matrices by uppercase letters. The trace and Hermitian of a matrix are shown as $\text{Tr}(\cdot)$ and $(\cdot)^{\rm{H}}$, respectively. The Ferobenius inner product is defined as $\langle \bm{A}, \bm{B} \rangle_{F}= \text{Tr}(\bm{A}\bm{B}^{\rm{H}})$. 
$\| .\|$ denote the spectral norm and $ \bm{X} \succcurlyeq 0$ means that $ \bm{X} $ is a semidefinite matrix.  
%
We describe the linear operator $\mathcal{A}:\mathbb{R}^{m \times n} \rightarrow \mathbb{R}^{p}$ as $$\mathcal{A}(\bm{X}) =[\langle \bm{X},\bm{A}_{1}\rangle _{F} ,\dots, \langle \bm{X},\bm{A}_{p}\rangle _{F}]^{\rm T} $$ where $\bm{A}_{i} \in \mathbb{R}^{m \times n }$. The adjoint operator of $\mathcal{A}$ is defined as $ \mathcal{A}^{*}\bm{y} = \sum_{i=1}^{p}y_{i}\bm{A}_{i}$ and $ \mathcal{I} $ is the identity linear operator i.e. $\mathcal{I}\bm{X}= \bm{X}$. 

The orthogonal projection matrices onto the subspaces $\bm{\mathcal{U}}$ and $\bm{\mathcal{U}}^{\perp}$ are shown by $\bm{P}_{\bm{\mathcal{U}}} := \bm{{U}}\bm{{U}}^{\rm{H}},$ and $\bm{P}_{\bm{\mathcal{U}}^{\perp}} := \bm{I} - \bm{P}_{\bm{\mathcal{U}}},$
where $\bm{U}$ is a basis for the subspace $\bm{\mathcal{U}}$ and $\bm{I}$ is the identity matrix.

\section{Single Weight Nuclear Norm Minimization}\label{section.Existing results} 
{\color{\change}This section is provided in order to show the performance of the single-weighted strategy provided in \cite{eftekhari2018weighted} and to highlight the amount of improvements compared to the regular nuclear norm minimization in the presence of prior subspace information.} Recht et al. in \cite{recht2010guaranteed} show that the nuclear norm minimization \eqref{minnuclear} robustly recover $ \bm{X} $ with noisy measurements as long as the linear operator $  \mathcal{A} $ satisfies the RIP condition defined below.

\begin{defn}
	For constant $\delta_{r}(\mathcal{A})\in (0,1]$, a linear operator $\mathcal{A}$ satisfies RIP condition if
	\begin{align}\label{eq.RIP_A}
		(1-\delta_{r}(\mathcal{A}))\|\bm{X}\|_{F} \leq \|\mathcal{A}(\bm{X})\|_{2} \leq (1+\delta_{r}(\mathcal{A}))\|\bm{X}\|_{F}
	\end{align}
	holds for every $\bm{X}$ with $ {\rm{rank}}(\bm{X}) \leq r$.
\end{defn} 
Almost all linear operators satisfy RIP condition if the number of measurements is sufficiently large. For example, a linear operator with independent  Gaussian entries with zero-mean and variance $ 1/p $ satisfies RIP condition with high probability when $ p \ge rn \log n/\delta^2_{r}(\mathcal{A}) $.

First, we explain principal angles between subspaces $\bm{\mathcal{U}}$  and $\widetilde{\bm{\mathcal{U}}}  $ with $r:={\rm dim}(\bm{\mathcal{U}}) \le {\rm dim}(\widetilde{\bm{\mathcal{U}}})=:r'	$. There are $ r $ non-increasing principal angles $\mathbf{\bm{\theta}}_{u}  \in [0^{\degree} , 90^{\degree}]^{r}$
\begin{align}
	&\mathbf{\theta}_{u}(i) =\min \Big\{ \cos^{-1} \left( \frac{|\langle \bm{u} , \widetilde{\bm{u}} \rangle|}{\|\bm{u}\|_2 \|\widetilde{\bm{u}} \|_2} \right) \quad : \quad ~ \bm{u} \in \bm{\mathcal{U}}, ~ \widetilde{\bm{u}} \in \widetilde{\bm{\mathcal{U}}} ,\nonumber \\
	&\bm{u} \perp \bm{u}_j, ~ \widetilde{\bm{u}} \perp \widetilde{\bm{u}}_j  ~ : \quad  \forall j \in \{i+1,\dots,r\} \Big\}
\end{align} 
where $ \bm{u}  $ and $ \widetilde{\bm{u}} $ are called principal vectors and the maximum principal angle is denoted by $ \mathbf{\theta}_{u}(1)$ \cite{daei2018optimal}. 
\begin{thm} \cite{eftekhari2018weighted} \label{theoremEftekhari}
	Let $\bm{X}_r = \bm{U}_{r}\bm{\Sigma}_{r}\bm{V}_{r}^{\rm H} \in \mathbb{R}^{n \times n }$ for an integer $r \le n $ be a truncated SVD from $\bm{X} \in \mathbb{R}^{n \times n}$ and consider the residual $\bm{X} _{r^+} = \bm{X} - \bm{X}_r $. Suppose that $\widetilde{\bm{\mathcal{U}}}_{r} $ and $\widetilde{\bm{\mathcal{V}}}_{r}$ are prior subspace information about $\bm{\mathcal{U}}_{r}= {\rm span}(\bm{X}_r)$ and  $\bm{\mathcal{V}}_{r}={\rm span}(\bm{X}_r^{\rm H})$. Assume that the linear operator $ \mathcal{A} $ satisfies RIP condition with 
	\begin{align}\label{RIP-CONDITION-1}
		\delta_{32r}(\mathcal{A}) \le \frac{0.9 - \max\{\alpha_1,\alpha_2\}/\sqrt{30}}{0.9 + \max\{\alpha_1,\alpha_2\}/\sqrt{30}}.
	\end{align}
	Then, for weights $ \lambda $ and $ \gamma $, the solution $ \widehat{\bm{X}} $ of \eqref{arminproblem} with noisy measurements $ \bm{y} = \mathcal{A}(\bm{X}+\bm{E})\in\mathbb{R}^{p}$ satisfies:
	\begin{align}
		\| \widehat{\bm{X}}  - \bm{X}\|_F \le \frac{\| \bm{X} _{r^+} \|_{*}}{\sqrt{r}} + e,
	\end{align}
	where $\|\mathcal{A}(\bm{E}) \|_2  \le e $ and $\alpha_1$, $\alpha_2$ are 
	\begin{align}
		& \alpha_1 := \nonumber \\ 
		&\sqrt{ \frac{\lambda^4 \cos^2 \mathbf{\theta}_{u}(1) + \sin^2 \mathbf{\theta}_{u}(1) }{\lambda^2 \cos^2 \mathbf{\theta}_{u}(1) + \sin^2 \mathbf{\theta}_{u}(1) }} + \sqrt{ \frac{\gamma^4 \cos^2 \mathbf{\theta}_{v}(1) + \sin^2 \mathbf{\theta}_{v}(1)}{\gamma^2 \cos^2 \mathbf{\theta}_{v}(1) + \sin^2 \mathbf{\theta}_{v}(1) }} \nonumber \\ 
		&\alpha_2 := \nonumber \\
		& \sqrt{ \frac{2 (1-\lambda^2)\sin^2 \mathbf{\theta}_{u}(1)}{\lambda^2 \cos^2 \mathbf{\theta}_{u}(1) + \sin^2 \mathbf{\theta}_{u}(1) }} + \sqrt{ \frac{2 (1-\gamma^2)\sin^2 \mathbf{\theta}_{v}(1)}{\gamma^2 \cos^2 \mathbf{\theta}_{v}(1) + \sin^2 \mathbf{\theta}_{v}(1) }}.
	\end{align}
\end{thm}

\begin{rem}
	\label{rem:eftekhateri}	For  $ \lambda = \gamma = 1 $, problem \eqref{arminproblem} is reduced to un-weighted nuclear norm minimization \eqref{minnuclear}, which results in  $\alpha_1 = 2$, $\alpha_2 = 0 $ and $\delta_{32r}(\mathcal{A}) \le 0.42 $. This is more conservative than the result of \cite{recht2010guaranteed}, $ \delta_{5r}(\mathcal{A}) \le 0.1 $,  because $\delta_{5r}(\mathcal{A}) \le 0.05$ implies  $\delta_{32r}(\mathcal{A}) \le 0.42 $.
\end{rem}
\section{Non-Uniform Weighting}\label{ourwork}
In this section, we generalize the weighted nuclear norm minimization theory of \cite{eftekhari2018weighted} to the non-uniform weights. 
Suppose that $ \widetilde{\bm{\mathcal{U}}}_{r^{\prime}} $ and $\widetilde{\bm{\mathcal{V}}}_{r^{\prime}}$ are prior subspace information forming angles with $\bm{\mathcal{U}}_{r}$ and $\bm{\mathcal{V}}_{r}$, respectively. 
We determine the optimal weights according to the principal angles.

Our main result in Theorem \ref{our theorem} provides recovery guarantees for noisy and noiseless measurements. It also covers the uniformly weighted nuclear norm minimization.  We show that the RIP condition in non-uniformly weighted case is weaker than the uniform case. 
\begin{thm}\label{our theorem}
	Let $\bm{X}_r \in \mathbb{R}^{n \times n}$ be a rank $ r $ truncated SVD of $\bm{X}$ and $\bm{X}_{r^+} = \bm{X}-\bm{X}_r $. Also, $\bm{\mathcal{U}}_{r}= {\rm span}(\bm{X}_r)$ and $\bm{\mathcal{V}}_{r}={\rm span}(\bm{X}_r^{\rm H})$ denote the column and row subspaces of $\bm{X}_r$, with their corresponding prior subspace information denoted by $\widetilde{\bm{\mathcal{U}}}_{r^{\prime}}$ and $\widetilde{\bm{\mathcal{V}}}_{r^{\prime}}$, which are $ r^{\prime} $- dimensional subspaces. For each pair of subspaces consider the non-increasing angle vector
	\begin{align*}
		\mathbf{\bm{\theta}}_{u}= \angle[\bm{\mathcal{U}}_r,\widetilde{\bm{\mathcal{U}}}_r],~ \mathbf{\bm{\theta}}_{v}=\angle[\bm{\mathcal{V}}_r,\widetilde{\bm{\mathcal{V}}}_r],
	\end{align*} 
	which represent the accuracy of prior information.
	
	Suppose that the linear operator $ \mathcal{A}$ satisfies the RIP condition:	
	\begin{align}\label{our RIP condition}
		\delta_{32r}(\mathcal{A}) \le	\frac{1-\sqrt{\frac{1}{15} ( \alpha_3^2 + \alpha_4^2) }}{ 1 + \sqrt{ \frac{1}{15} (\alpha_3^2 + \alpha_4^2)}} 
	\end{align}
	where 
	\begin{align}
		&\alpha_3 := \max_i \sqrt{ \frac{\lambda_{1}^4(i) \cos^2 \mathbf{\theta}_{u}(i) + \sin^2 \mathbf{\theta}_{u}(i)}{\lambda_{1}^2(i) \cos^2 \mathbf{\theta}_{u}(i) + \sin^2 \mathbf{\theta}_{u}(i) }} + \nonumber \\
		& \label{eq:a3} \quad \quad \quad \quad \quad \max_i  \sqrt{ \frac{\gamma_{1}^4(i) \cos^2 \mathbf{\theta}_{v}(i) + \sin^2 \mathbf{\theta}_{v}(i)}{\gamma_{1}^2(i) \cos^2 \mathbf{\theta}_{v}(i) + \sin^2 \mathbf{\theta}_{v}(i) }} \\
		& \label{eq:a4} \alpha_4:= \max_i \sqrt{{d}_i(\bm{\theta}_{u}, \bm \lambda_{1},\bm \lambda_{2}) } + \max_i \sqrt{ {d}_i(\bm {\theta}_{v}, \bm \gamma_{1}, \bm \gamma_{2})} \\
		& {d}_1(\bm{\theta}, \bm a, \bm b) := \max_i \Big(\Big(\frac{a(i)}{\sqrt{a^2(i)\cos^2 \theta(i) + \sin^2 \theta(i)}}-1 \Big)^2  \nonumber \\
		& \label{eq:d1} \quad \quad + \frac{(1-a(i))^2\cos^2 \theta(i)\sin^2 \theta(i)}{a^2(i)\cos^2 \theta(i) + \sin^2 \theta(i)} \Big)  \\ 
		& \label{eq:d2} {d}_2(\bm{\theta}, \bm a, \bm b) := \max_i \, (b(i)-1)^2.
	\end{align}
	Then for the solution  $\widehat{\bm{X}}$  of \eqref{ourproblem} we have:
	\begin{align}\label{error bound}
		\|\widehat{\bm{X}} - \bm{X} \|_{F} \le  C_0\|\bm{X}_{r^{+}}\|_{*} + C_1e 
	\end{align}
	where 
	\begin{align}\label{cte}
		& C_0 := \frac{ \frac{ 4 }{ ( 1- \delta_{32r}(\mathcal{A})) \sqrt{30r}} }{ 1-\frac{1+ \delta_{32r}(\mathcal{A})}{1-\delta_{32r}(\mathcal{A})} \sqrt{ \frac{1}{15} ( \alpha_3^2 + \alpha_4^2 )}	} \nonumber \\ 
		& C_1:= \frac{ \frac{2}{ 1-\delta_{32r}(\mathcal{A})} \left( 1 + \sqrt{ \frac{1}{15} ( \alpha_3^2 + \alpha_4^2 )} \right) }{ 1-\frac{1+ \delta_{32r}(\mathcal{A})}{1-\delta_{32r}(\mathcal{A})} \sqrt{ \frac{1}{15} ( \alpha_3^2 + \alpha_4^2 )}	}.
	\end{align}
\end{thm}
\begin{proof}
	See Appendix \ref{proof theorem 2}.
\end{proof}
\begin{rem}
	If $\bm{\Lambda} = \bm{\Gamma} = \bm{I}_{r^{\prime}}$ then $\bm{Q}_{\widetilde{\bm{\mathcal{U}}}_{r^{\prime}}} = \bm{Q}_{\widetilde{\bm{\mathcal{V}}}_{r^{\prime}}} = \bm{I}_{n}$ and the problem reduces to the standard nuclear norm minimization in \cite{recht2010guaranteed}. Also, if  $\bm{\Lambda} = \lambda\bm{I}_{r^{\prime}}$ and $\bm{\Gamma} = \gamma\bm{I}_{r^{\prime}}$ then problem \eqref{ourproblem} reduces to \eqref{arminproblem} which is studied in \cite{eftekhari2018weighted}.
\end{rem}
\begin{rem}
	Our goal is to weaken the RIP condition in \eqref{our RIP condition}. Therefore, we choose weights that minimize $\sqrt{\alpha_3^2+\alpha_4^2}$. For non-increasing principal angles, the weights should also be non-increasing. Since the error bound in \eqref{error bound} depends on weights so another approach is to obtain optimal weights to minimize $C_0$ in \eqref{error bound}. 
\end{rem}

\begin{rem}\label{compare RIP}
	The RIP condition in \eqref{our RIP condition} is weaker than those in the single weight and the unweighted nuclear norm minimization due to the availability of higher degrees of freedom in terms of principal angles. In Table \ref{compre RIP table}, a  numerical comparison of RIP conditions for uniformly and non-uniformly weighted, and standard nuclear norm minimization is presented. Different scenarios of accurate and inaccurate subspace estimators are considered where we obtain the optimal weights by maximizing the RIP bound \eqref{our RIP condition}. As we expected, the RIP condition by using non-uniform weights is weaker than unweighted and single-weight scenarios. But for the standard problem, the RIP condition in \eqref{our RIP condition} and \eqref{RIP-CONDITION-1} is slightly more conservative than $\delta_{5r}(\mathcal{A}) \le 0.1$ in \cite{recht2010guaranteed} (see Remark \ref{rem:eftekhateri}).
\end{rem}
\begin{table*}
	\centering
	\caption{Comparison of RIP condition for the random matrix with parameters $ n = 30$, $ r = 3$ and $ r^{\prime} = 7 $}
	\scalebox{0.65}{
		\setlength\extrarowheight{3pt}
		\begin{tabular}{c|c|c|c|c|c|c}
			\hline
			$ \mathbf{\bm{\theta}}_{u} $ & $ \mathbf{\bm{\theta}}_{v} $ & $ \delta_{32r}(\mathcal{A})-{\rm Standard} \eqref{our RIP condition} $ & $ \delta_{32r}(\mathcal{A})-{\rm uniform ~ weight} \eqref{our RIP condition}$ & $  \delta_{32r}(\mathcal{A})-{\rm non-uniform ~ weight} \eqref{our RIP condition} $ & $ \delta_{32r}(\mathcal{A})-{\rm uniform ~ weight}\eqref{RIP-CONDITION-1} $ &$ \delta_{32r}(\mathcal{A})-{\rm standard} \eqref{RIP-CONDITION-1}$\\
			\hline 
			$ [2.26 , 2.98 , 3.10] $ & $ [1.91 , 2.87 , 3.40] $ & $ 0.32 $ & $ 0.46 $ & $  0.68 $ & $ 0.58 $ &$ 0.1 $ \\
			\hline
			$ [23.1 , 24.54 , 27.56] $ & $ [20.95 , 20.06 , 34.03] $ & $ 0.32  $ & $ 0.35$ & $ 0.39 $ & $ 0.20 $ &$ 0.1 $ \\
			\hline
			$ [2.10 , 21.39 , 27.07] $ & $ [3.49 , 18.17 , 24.68] $ & $ 0.32 $ & $ 0.36 $ & $ 0.40 $ & $ 0.22 $ &$  0.1 $\\
			\hline 
			$ [50.31 , 58.63 , 68.75] $ & $ [54.36 , 66.41 , 72.14]$ & $ 0.32 $ & $ 0.30 $ & $  0.32  $ & $ 0.11 $ &$ 0.1 $\\
			\hline 
	\end{tabular}}
	\label{compre RIP table}
\end{table*}

\section{FDD massive MIMO}\label{sec.fdd}
In this section, we provide the well-known system model used in FDD massive MIMO \cite{FDD},\cite{love2008overview},\cite{shen2016compressed} and illustrate how our method can be applied to this application. {\color{\change}Consider a fixed BS with $  M $ antennas and $ K $ single-antenna moving users with associated Doppler frequencies $\bm{\nu}=[\nu_1,..., \nu_K]^T$.} In FDD systems, the BS first sends a few pilots to the users, then the users estimate their own channels and feed back the estimates to the BS. {\color{\change}The channel between BS and $k$-th user is assumed quasi-statistic during $T$ time blocks and is described as \cite{liu2020angular,li2019time}:
	\begin{align}
		\label{2-fdd}
		\bm{h}_{k}=\sum_{l=1}^r\alpha_{k,l}{\rm e}^{j2\pi \nu_k T }\bm{a}(\theta_l) \in\mathbb{C}^{M},
\end{align} }
where $ r $ is the number of propagation paths, $\alpha_{k,l}$ is $l$-th complex channel amplitude, $ \theta_l $ is the angle-of-departure (AoD) of the l-th path and $\bm{a}(\theta):=[1, e^{-j2\pi \frac{D}{\lambda}\cos(\theta)},\dots,  e^{-j2\pi \frac{D}{\lambda}(M-1)\cos(\theta)}]^T$ is the steering vector, where $ D $ and $ \lambda $ are the antenna spacing at the BS and carrier wavelength, respectively. The number of paths i.e. $r$ is often very fewer than the number of BS antennas i.e. $M$. Concatenating $\bm{h}_k$s leads to the MIMO channel matrix:{\color{\change}
	\begin{align}
		\bm{H} :=  \begin{bmatrix}
			\bm{h}_1, \dots , \bm{h}_K
		\end{bmatrix}= \bm{A}\bm{G}(\bm{\nu})\in \mathbb{C}^{M \times K},
\end{align}}
{\color{\change}where $ \bm{G}(\bm{\nu}) \in \mathbb{C}^{r \times K} $ is a function of Doppler frequencies of all users $\bm{v}:=[\nu_1,...,\nu_K]$ with $G(l,k) = \alpha_{l,k} {\rm e}^{j2\pi \nu_k T}$, and}  $ \bm{A} =\begin{bmatrix} \bm{a}(\theta_1), \dots , \bm{a}(\theta_r)
\end{bmatrix} \in \mathbb{C}^{M\times r}$. We know that $ {\rm rank(\bm{H})} \le \min\{ {\rm rank}(\bm{G}) ,  {\rm rank}(\bm{A})\} $ or  $ {\rm rank}(\bm{H}) \le \min\{ M,K,r\}  $. In massive MIMO systems with massive users, the number of contributing communication paths  is much smaller than $  M $ and $ K $. Thus,  $ {\rm rank}(\bm{H})\le r  $  and $\bm{H}$ is a low-rank matrix.
After $T$-th time blocks, the received signal at the $k$-th user can be expressed as
\begin{align}\label{1-fdd}
	\bm{y}_k  = \bm{\Phi} \bm{h}_k+\bm{e}\in \mathbb{C}^{T \times 1}
\end{align}
where $ \bm{\Phi}\in\mathbb{C}^{T\times M} $
is the pilot matrix transmitted during $T$ channel uses and $\bm{e}$ is the additive Gaussian noise. It is assumed that $T$ is less than the coherence interval so that the channel is invariant during $T$ time blocks. By writing the latter equation for all $K$ users in a matrix form, we have
\begin{align}
	\label{3-fdd}
	\bm{Y} = \bm{\Phi}\bm{H}+\bm{E} \in\mathbb{C}^{T\times K},
\end{align} 
where $ \bm{Y} = \begin{bmatrix}
	\bm{y}_1^{T}, \dots , \bm{y}_K^{T}
\end{bmatrix}^T  \in \mathbb{C}^{T \times K}$. After reformulation, finding an estimate for the low-rank matrix $\bm{H}$ from $\bm{Y}$ can be recast as solving \eqref{minnuclear}. {\color{\change}To estimate the prior knowledge, consider the channel matrix at two consecutive coherence intervals where users have different Doppler frequencies:
	\begin{align}
		\bm{H}^c=\bm{A}\bm{G}_c(\bm{\nu}^c), \bm{H}^p=\bm{A}\bm{G}_p(\bm{\nu}^p)     
	\end{align}
	where $\bm{H}^c$ and $\bm{H}^p$ correspond to channel matrices in the current and previous coherence intervals, respectively. $\bm{G}_c(\bm{\nu}^c)$ and $\bm{G}_p(\bm{\nu}^p)$ are the coefficients' matrices corresponding to the current and previous coherent intervals. Taking SVD of the channel matrices in the previous coherence interval provides column and row prior subspaces. By knowing the velocity of users in the previous and current time blocks, the associated Doppler frequencies $\bm{\nu}^c$ and $\bm{\nu}^p$ can be obtained which then provide an estimate of the principal angles between $\bm{H}^c$ and $\bm{H}^p$. Then, by solving \eqref{ourproblem} with the weights proposed in Section \ref{ourwork}, the required number of pilots i.e. $T$ to estimate the current channel matrix can be substantially decreased which in turn leads to a huge resource saving. It is worth mentioning that there are also several papers (e.g. see \cite{liu2020angular,qin2018sparse,ma2018sparse,li2019time}) which exploit a time-varying channel adopting the auto regressive (AR) model which indeed differs from our modeling here. Specifically, these works do not indeed promote the inherent features of the channels in order to decrease the number of training overhead but instead provide estimates of the AR parameters using maximum likelihood (ML) estimation which is done by expectation maximization (EM) algorithm.}
\section{Simulation Results}\label{section.simulation}
In this section, we provide numerical experiments to show that non-uniform weighting strategy perform better that uniform weighting strategy. All experiments are performed using CVX package and numerical optimization is used to obtain optimal weights. 
\subsection{Numerical experiments}
$\bm{X} \in \mathbb{R}^{n\times n} $ is a square matrix with $ n = 20 $ and $ r=3 $. We use $\bm{X}^{\prime} = \bm{X}+ \bm{N} $ with $\bm{N}$ a small random perturbation matrix to construct the prior subspaces  $ \widetilde{\bm{\mathcal{U}}}_{r^{\prime}}$ and $ \widetilde{\bm{\mathcal{V}}}_{r^{\prime}}$ as spans of $ \bm{X}^{\prime} $ and $ \bm{X}'^{\rm H}$. Also $ \bm{\mathcal{U}}_{r}  $ and  $\bm{\mathcal{V}}_{r}  $ subspaces have known principal angles $   \mathbf{\bm{\theta}}_{u}  \in [0^{\degree} , 90^{\degree}]^{r} $ and $\mathbf{\bm{\theta}}_{v}  \in [0^{\degree} , 90^{\degree}]^{r} $ with $ \widetilde{\bm{\mathcal{U}}}_{r^{\prime}} $ and $\widetilde{\bm{\mathcal{V}}}_{r^{\prime}} $, respectively. $ \bm{U}_r$ and $\widetilde{\bm{U}}_{r^{\prime}}$  without loss of generality can be chosen such that 
\begin{align*}
	\bm{U}_r^{\rm H} \widetilde{\bm{U}}_{r^{\prime}}= [\cos \mathbf{\bm{\theta}}_{u} \quad \bm{0}_{r \times r^{\prime}-r} ],\bm{V}_r^{\rm H} \widetilde{\bm{V}}_{r^{\prime}}= [\cos \mathbf{\bm{\theta}}_{v} \quad \bm{0}_{r \times r^{\prime}-r} ],
\end{align*}
which amounts to redefining $  \bm{U}_r $ and $ \widetilde{\bm{U}}_{r^{\prime}} $ as the left and right singular matrices of $ \bm{U}_r^{\rm H} \widetilde{\bm{U}}_{r^{\prime}} $. The same conclusion can be cast for $\bm{V}_r$ and $\widetilde{\bm{V}}_{r^{\prime}}$.

We compare the problems \eqref{ourproblem}, \eqref{arminproblem} and the standard nuclear norm with optimal weights in different $\mathbf{\bm{\theta}}_{u}$ and $\mathbf{\bm{\theta}}_{v}$. We repeat each experiment 50 times with different choices of $\mathcal{A}$ and noise in noisy problems. For the solution $ \widehat{\bm{X}}$ of the problem, we evaluate the normalized recovery error (NRE) defined as:
$$ {\rm NRE} := \frac{\| \widehat{\bm{X}} - \bm{X}\|_{F}}{\|\bm{X}\|_{F}}.$$  An experiment is  successful if $ \rm {NRE} \le 10^{-4}$ .

Fig. \ref{fig1} shows the success rate and $\rm NRE$ without noise. In this experiment, we assume that the accuracy of prior information is good and the principal angles between subspaces are $ \mathbf{\bm{\theta}}_{u} = [1.0 , 1.6, 2.0, 2.2]$ and $ \mathbf{\bm{\theta}}_{v} = [1.0 , 1.4, 1.5, 2.8] $. We observe that weighted matrix recovery with non-uniform weights outperforms the single weight and standard problems.

\begin{figure*}[t]
	\centering
	\subfigure{\label{fig1:a}\includegraphics[width=3.5in, height=1.5in]{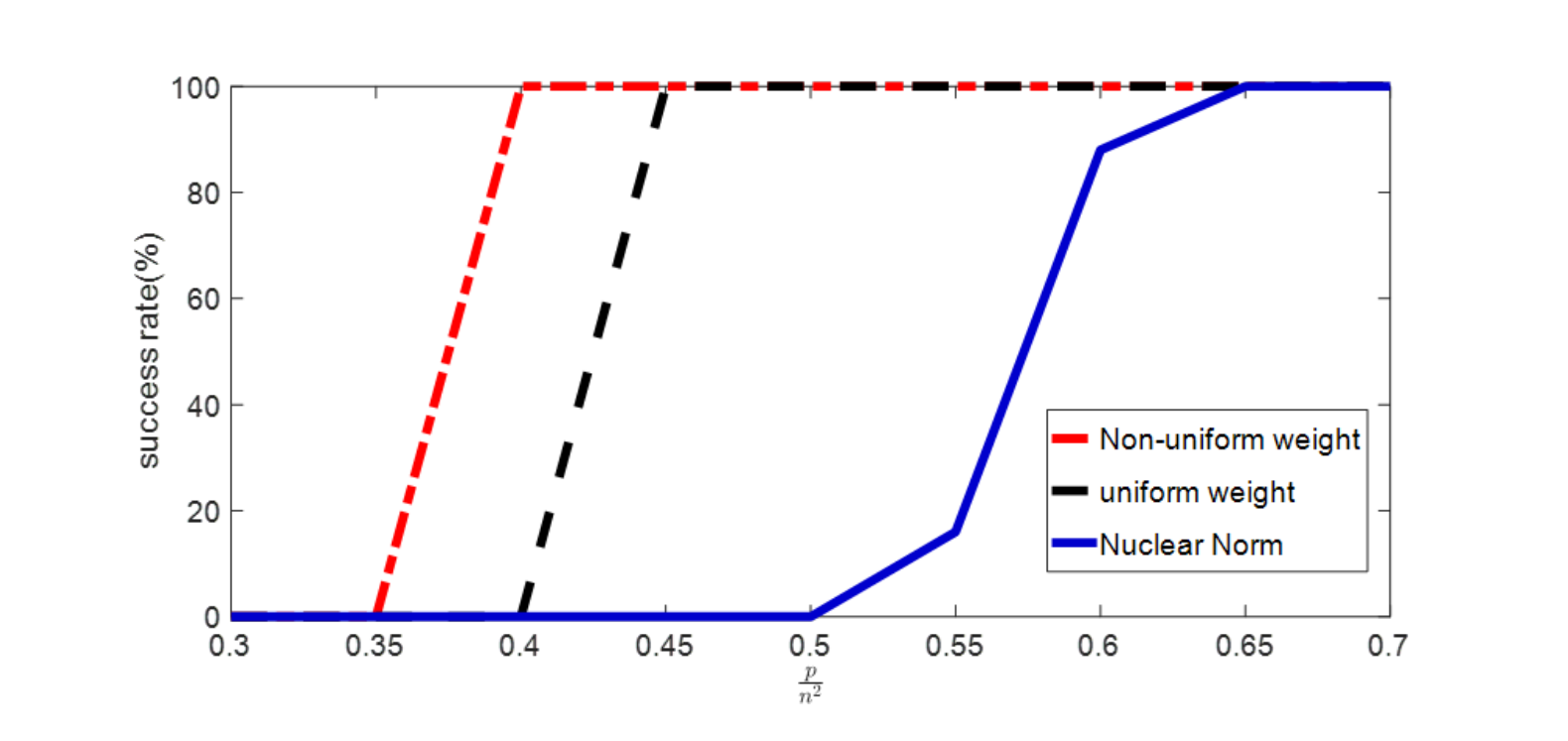}}
	\subfigure{\label{fig1:b}\includegraphics[width=3.5in, height=1.5in]{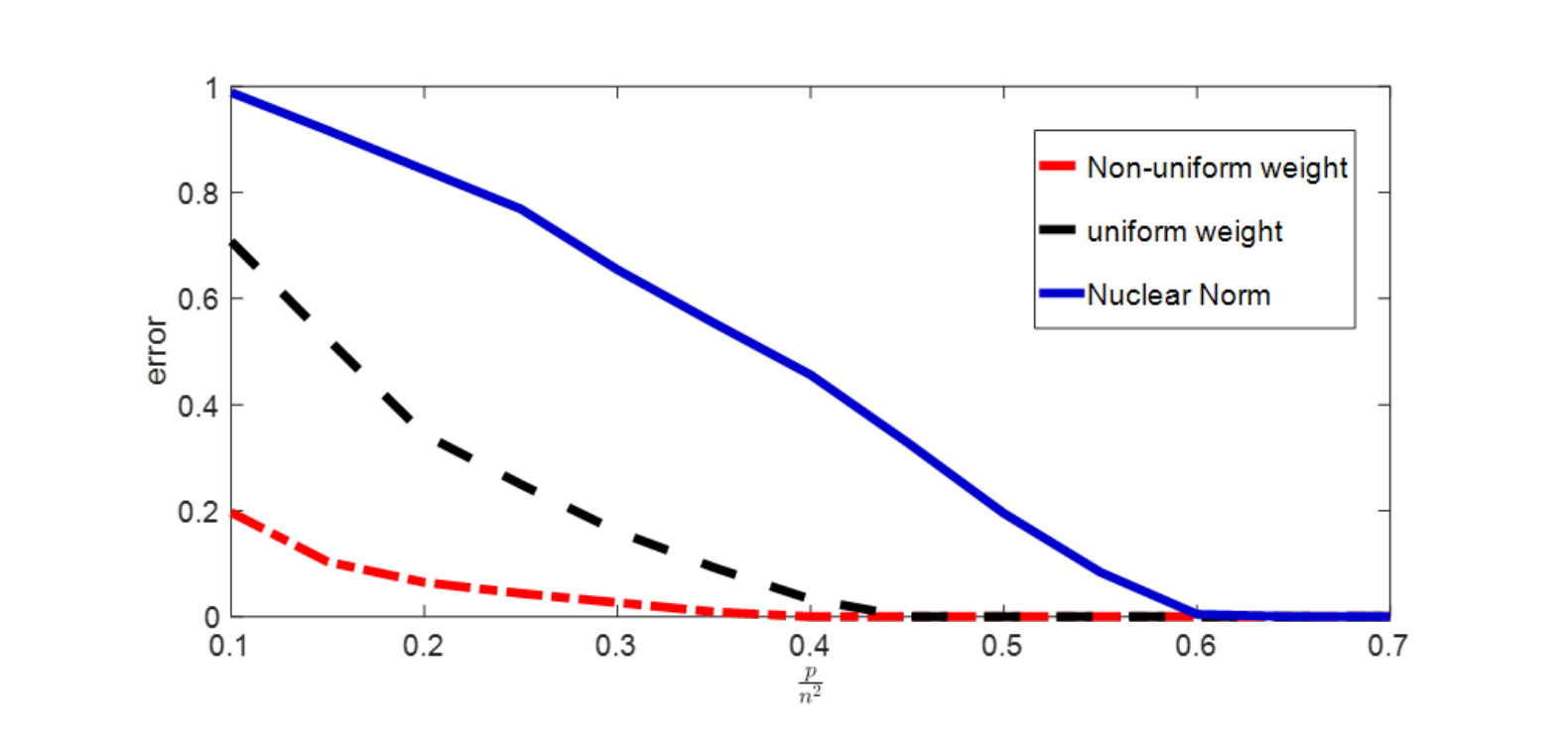}}
	\caption{Matrix recovery with different approaches without noise. Principal angles are $\mathbf{\bm{\theta}}_{u} = [1.0 , 1.6, 2.0, 2.2]$ and $\mathbf{\bm{\theta}}_{v} = [1.0 , 1.4, 1.5, 2.8] $.}
	\label{fig1}
\end{figure*}
In Fig. \ref{fig2}, the principal angles are $\mathbf{\bm{\theta}}_{u} = [2 , 13, 18, 27]$  and $\mathbf{\bm{\theta}}_{v} = [2 , 13, 18, 23]$ . In other words, some directions are accurate and some are not. As expected, the performance of matrix recovery with non-uniform weights is better than the other methods.
\begin{figure*}[t]
	\centering
	\subfigure{\label{fig2:a}\includegraphics[width=3.5in, height=1.5in]{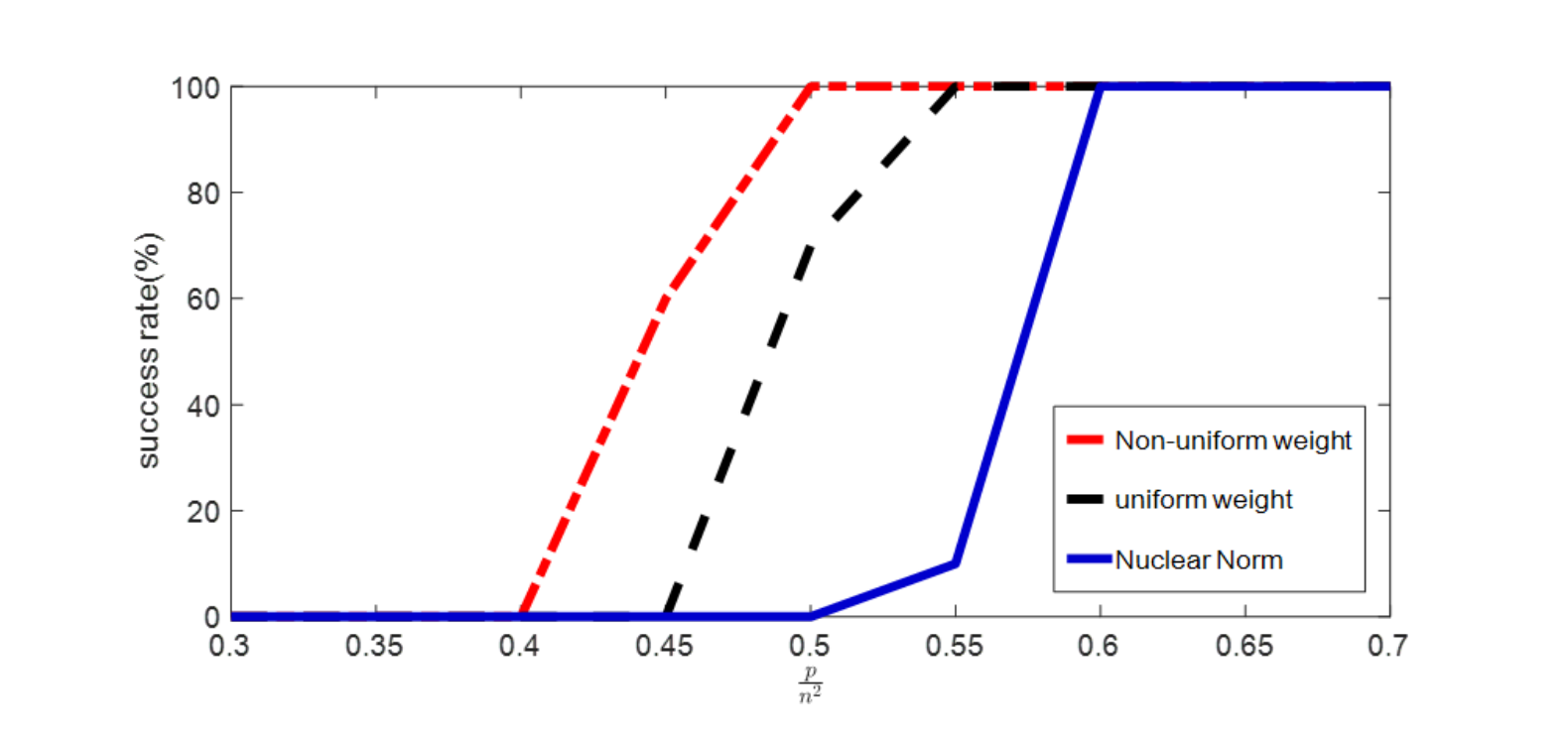}}
	\subfigure{\label{fig2:b}\includegraphics[width=3.5in, height=1.5in]{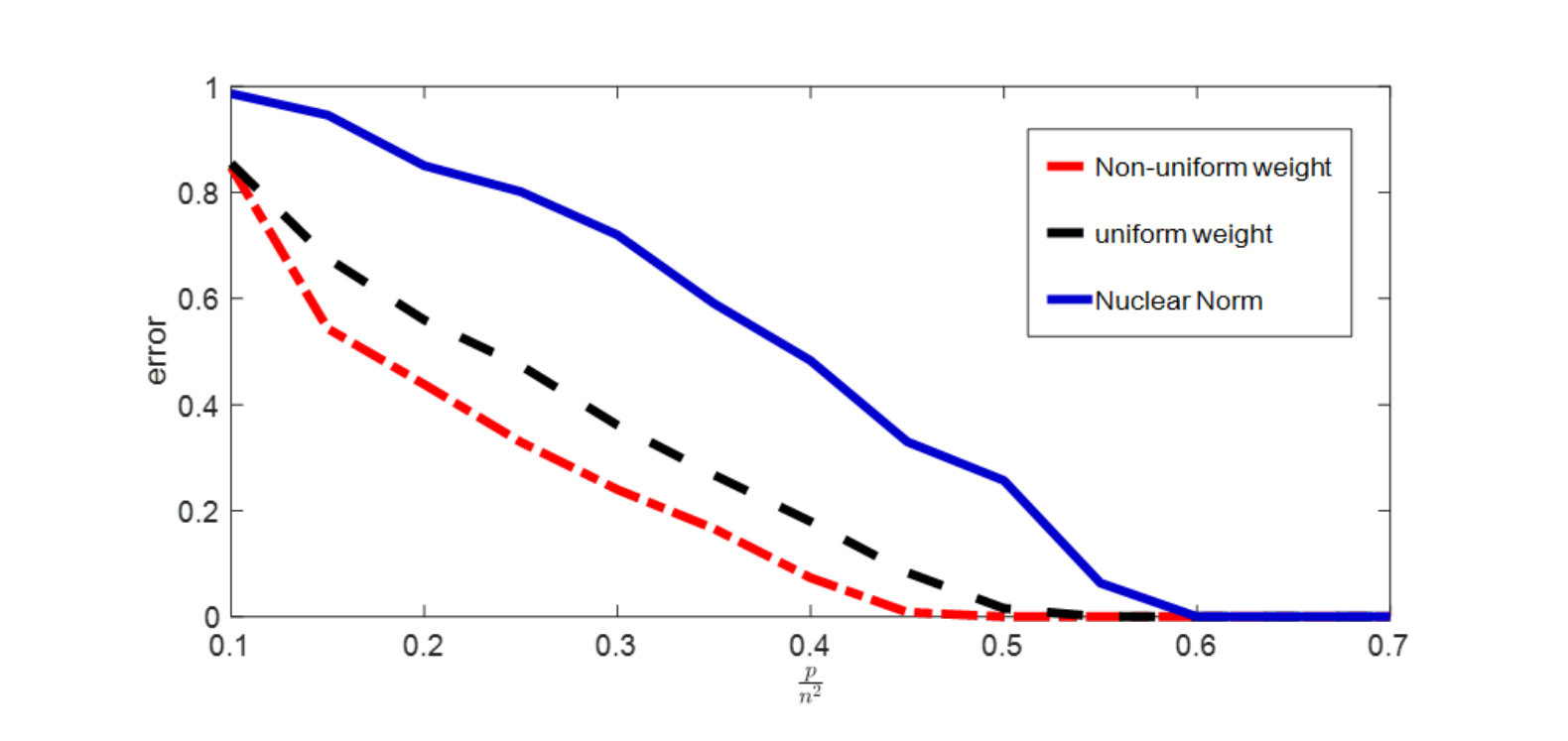}}
	\caption{Matrix recovery using different approaches without noise. Principal angles are $\mathbf{\bm{\theta}}_{u} = [2 , 13, 18, 27]$  and $\mathbf{\bm{\theta}}_{v} = [2 , 13, 18, 23]$.}
	\label{fig2}
\end{figure*}
In Fig. \ref{fig3}, the accuracy of prior information is not good: $\mathbf{\bm{\theta}}_{u} = [10 , 15, 19, 23]$ and  $\mathbf{\bm{\theta}}_{v} = [8 , 10, 15, 24]$.
\begin{figure*}[t]
	\centering
	\subfigure{\label{fig3:a}\includegraphics[width=3.5in, height=1.5in]{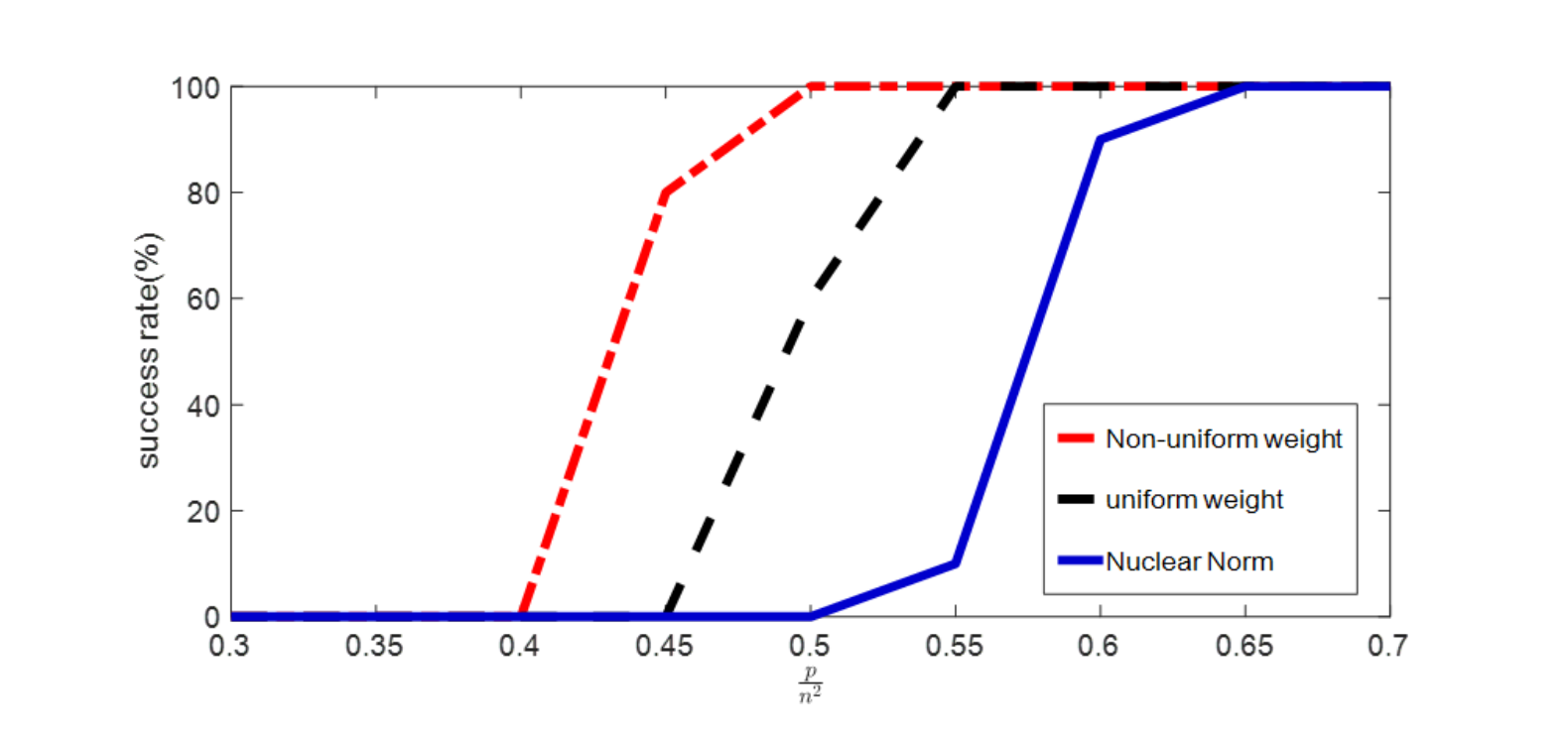}}
	\subfigure{\label{fig3:b}\includegraphics[width=3.5in, height=1.5in]{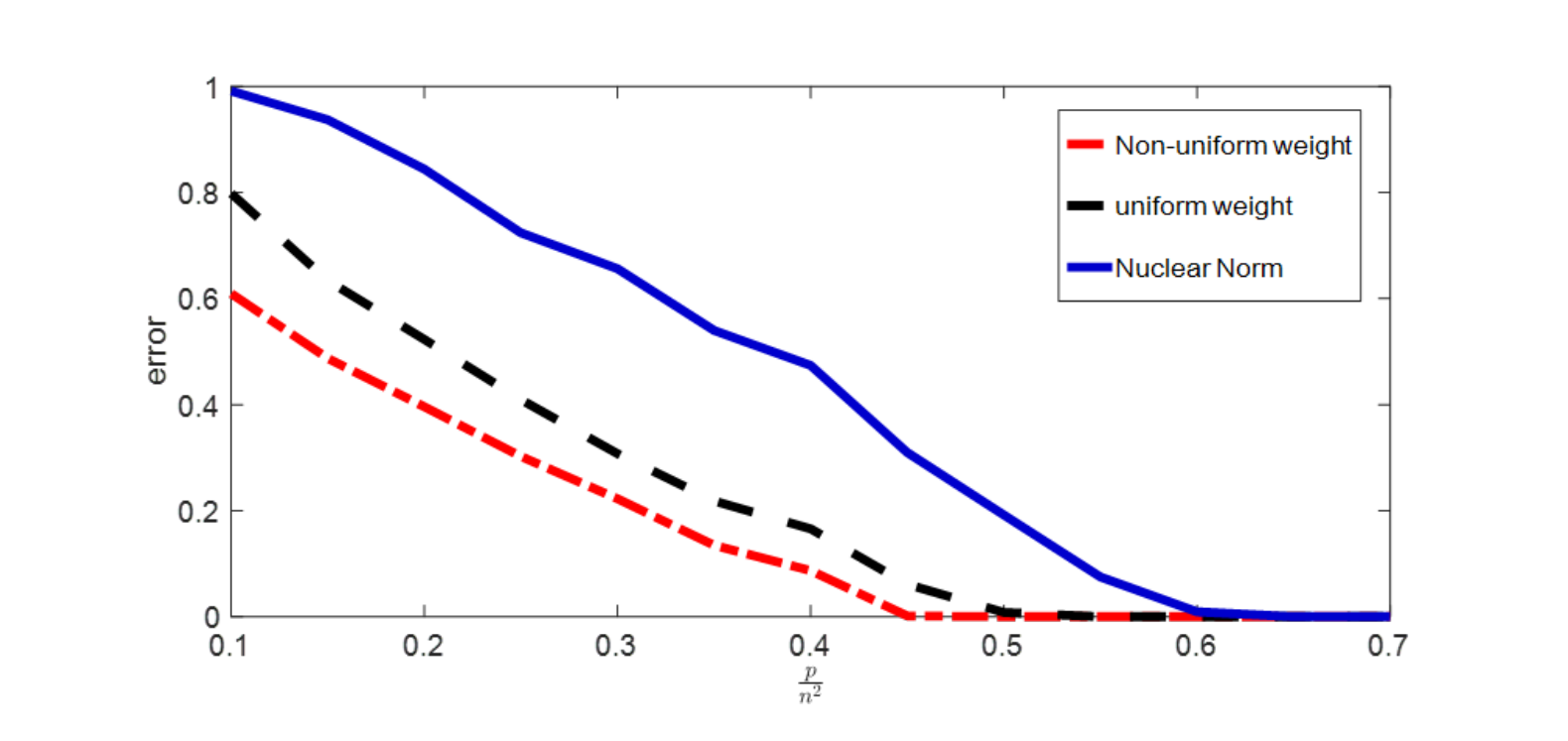}}
	\caption{Matrix recovery using different approaches without noise. Principal angles are $\mathbf{\bm{\theta}}_{u} = [10 , 15, 19, 23]$ and  $\mathbf{\bm{\theta}}_{v} = [8 , 10, 15, 24]$.}
	\label{fig3}
\end{figure*}
Fig. \ref{fig4} shows the $\rm NRE$ with noisy measurements for different accuracies. We observe that non-uniformly weighted matrix recovery is superior to the other methods in both noisy and noiseless cases.
\begin{figure*}[t]
	\centering
	\mbox{\subfigure{\includegraphics[width=2.34in, height=1.5in]{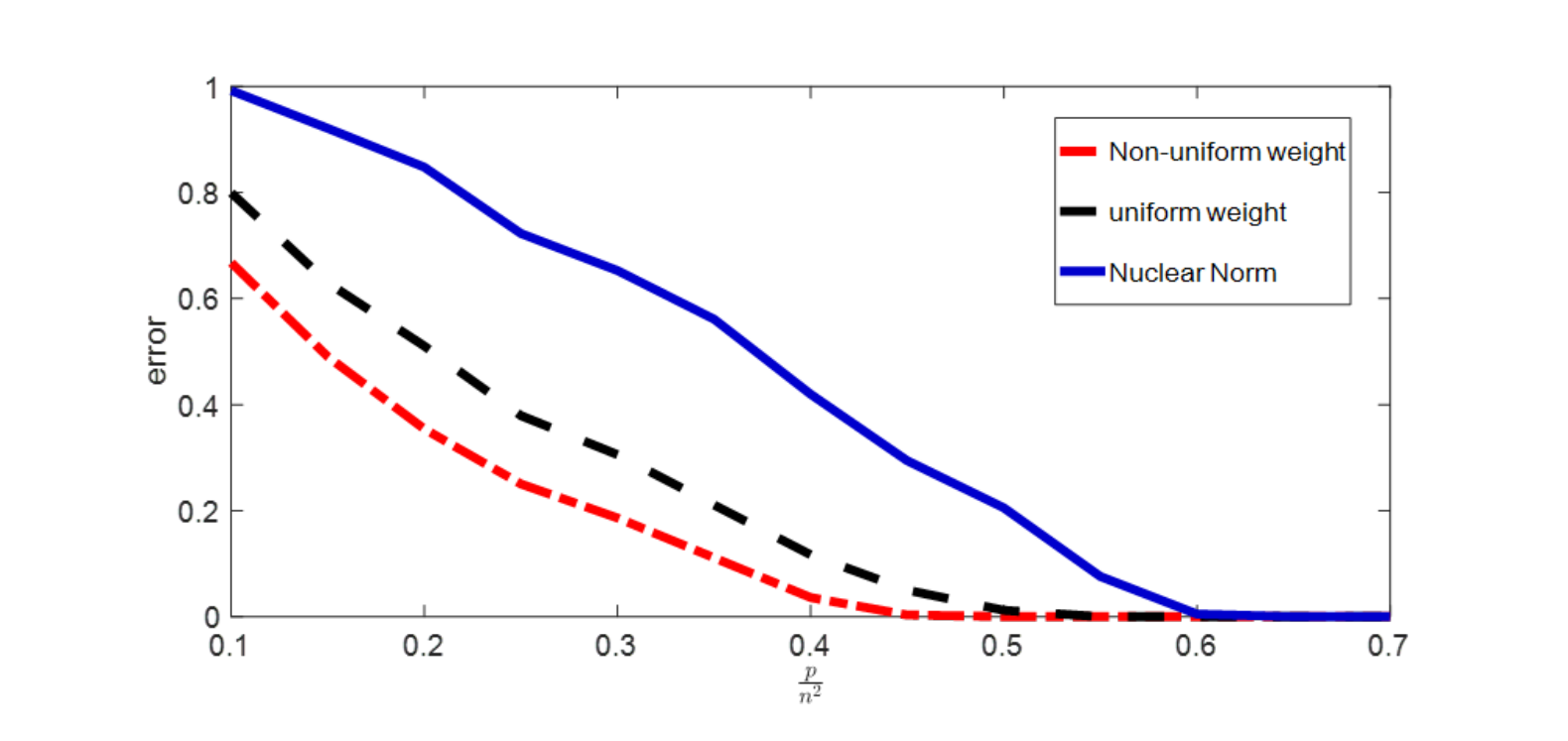}\label{fig4:a}}\quad
		\subfigure{\includegraphics[width=2.34in , height=1.5in]{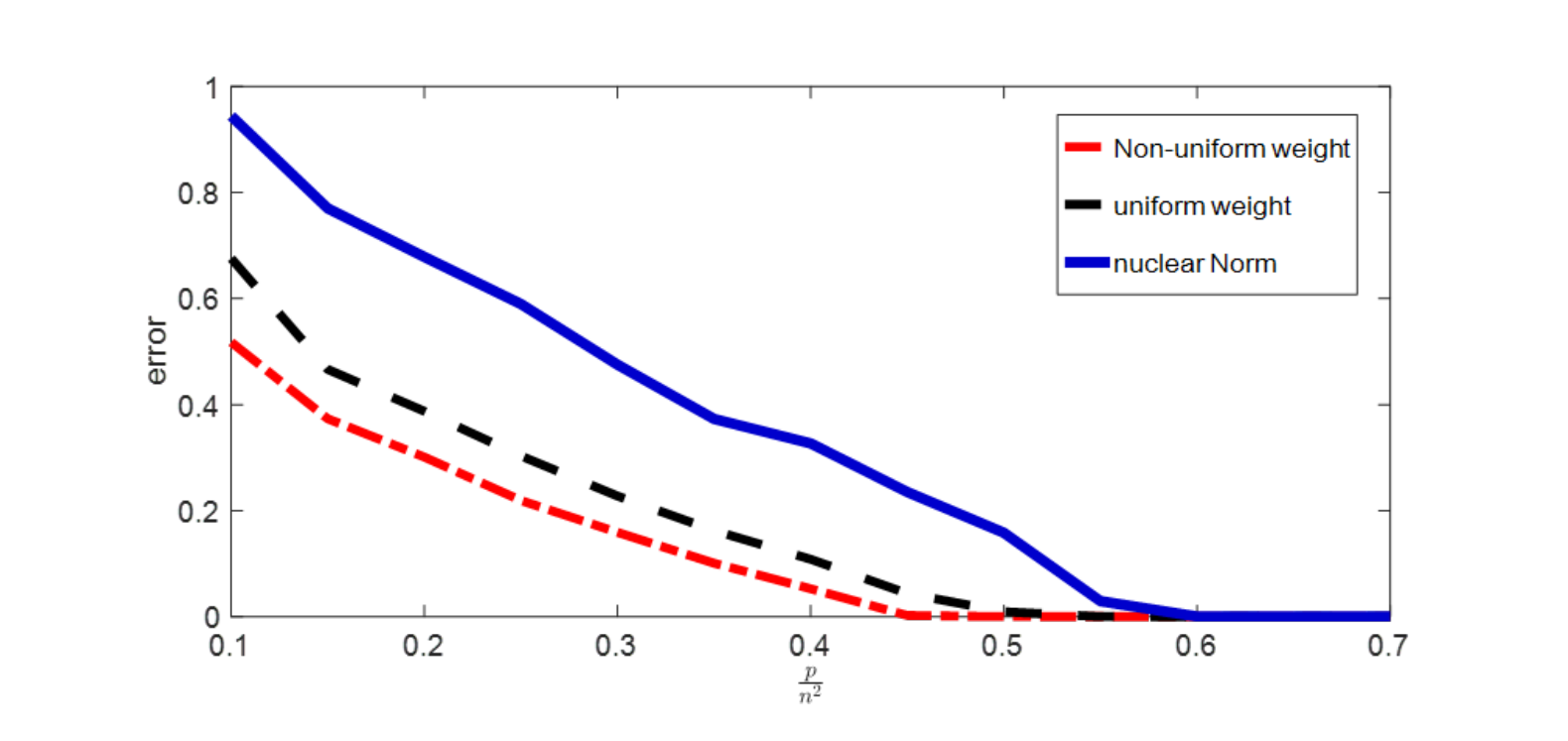}\label{fig4:b}}\quad
		\subfigure{\includegraphics[width=2.34in , height=1.5in]{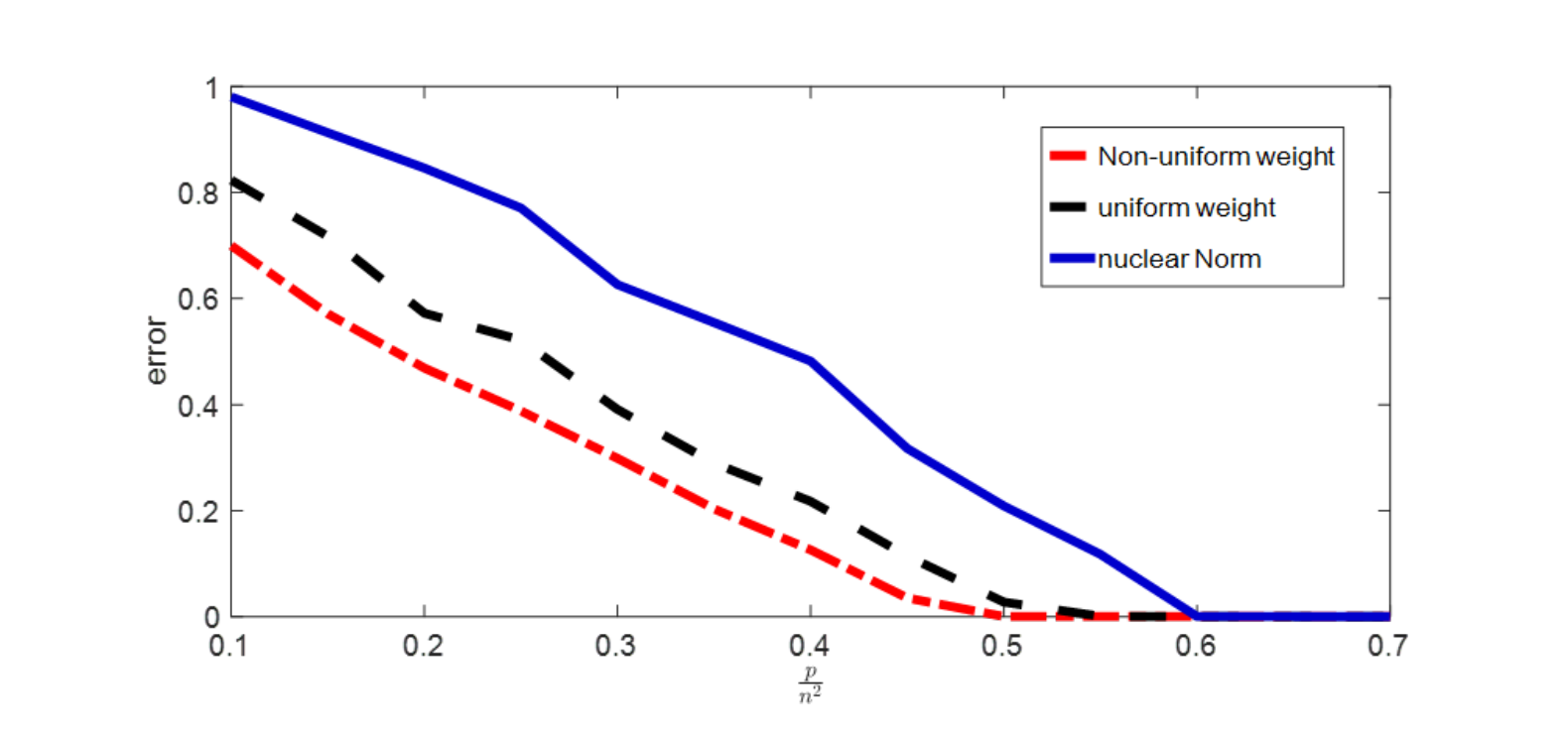}\label{fig4:c}}
	}	
	\caption{Noisy matrix recovery with principal angles in left, middle, and right similar to Figs. \ref{fig1}, \ref{fig2}, \ref{fig3}, respectively.}
	\label{fig4}
\end{figure*}
\subsection{Channel estimation in FDD massive MIMO}
We consider $ K=20 $ single-antenna users and a BS with $ M =64 $ antennas. The number of contributing communication paths in the current and previous coherence blocks are respectively considered as $ r = 6 $ and $ r' = 10 $. The
number of channel uses is fixed to $ T = 63 $. The distance between the BS antennas normalized by the wavelength is
$  \frac{D}{\lambda} =0.3 $. The elements of the pilot matrix $ \bm{\Phi} $ are drawn
from i.i.d. Gaussian distribution. The variance of additive
noise is 0.1. Table \ref{channel} shows the results of channel estimation
using nuclear norm minimization \eqref{minnuclear}, and \eqref{ourproblem} with uniform
and non-uniform weights. The NRE values corresponding to
each method are provided in Table  \ref{channel}. We observe that for a
fixed number of pilots (T), the performance of non-uniform
weights substantially outperforms the single-weight scenario \cite{eftekhari2018weighted} and nuclear norm minimization. This in turn implies that
to reach a fixed performance, the required number of pilots is
greatly reduced by our proposed method which could be also
translated as an enhanced spectral efficiency of the system
\begin{table*}
	\centering
	\caption{Time-varying channel estimation in FDD massive MIMO with parameters $M= 64$ , $ K= 20$, $ T=63 $, $ {\rm rank}(H)=R= 6 $ , $ R'= 12 $.}
	\scalebox{0.9}{
		\setlength\extrarowheight{3pt}
		\begin{tabular}{c|c|c|c}
			${\#}Time-slot $ & $ Non-uniform ~weight $ & $ Uniform~ weight $ & $ Nuclear ~Norm$ \\
			\hline 
			$ 1 $ & $ \bm{0.1162} $ & $0.1162  $ & $ 0.1162$ \\
			\hline
			$ 2 $ & $ \bm{0.0334}$ & $ 0.0665 $ & $ 0.1810 $  \\
			\hline
			$ 3 $ & $\bm{0.0179} $ & $ 0.0496$ & $0.1604 $  \\
			\hline 
			$ 4 $ & $ \bm{0.0029}$ & $ 0.0377$ & $ 0.1371 $ \\
			\hline 
			$ 5 $ & $ \bm{0.0022} $ & $ 0.0346$ & $ 0.1140 $  \\
			\hline 
	\end{tabular}}
	\label{channel}
\end{table*}

\section{conclusion}\label{conclusion}
In this paper, we proposed a weighted nuclear norm minimization to extract the feature of low-rankness in the presence of prior subspace information. The optimal weights are designed such that they weaken the RIP condition of the measurement operator as much as possible. Instead of using a single weight, multiple distinct weights are employed to penalize all the directions in the prior subspace. Analytical and simulation results show that our non-uniform weighting approach outperforms uniform weighting strategy and standard nuclear norm minimization. We have also provided a section that shows how our method can be applied to dynamic channel estimation in FDD massive MIMO.
\appendices
\section{Necessary Lemmas}
\label{section Necessary Lemmas}
In this section, we provide some necessary lemmas that are useful for proof of Theorem \ref{our theorem}.
\subsection{Constructing the Bases}
In this section, we introduce the bases which simplify the proofs.
\begin{lem}\label{Lem 3 }
	\cite{daei2018optimal} Suppose that $\bm{X}_r \in \mathbb{R}^{n \times n}$ is a rank $ r $ matrix with column and row subspaces  $\bm{\mathcal{U}}_{r}$ and $\bm{\mathcal{V}}_{r}$, respectively. Also, consider $ \widetilde{\bm{\mathcal{U}}}_{r^{\prime}} $ and $\widetilde{\bm{\mathcal{V}}}_{r^{\prime}}$ of dimension $ r^{\prime} \ge r $ with $r$ known principal angles $ \mathbf{\bm{\theta}}_{u}$ and $ \mathbf{\bm{\theta}}_{v}$ with subspaces $ \bm{\mathcal{U}}_{r} $ and $\bm{\mathcal{V}}_{r}$ as prior information. There exist orthogonal matrices $\bm{U}_r , \bm{V}_r \in \mathbb{R}^{n \times r }$ and  $\widetilde{\bm{U}}_{r^{\prime}} , \widetilde{\bm{V}}_{r^{\prime}}\in \mathbb{R}^{n \times r^{\prime}}$ such that 
	\begin{align*}
		& \bm{\mathcal{U}}_{r} =  {\rm span}(\bm{U}_r),  \quad \widetilde{\bm{\mathcal{U}}}_{r^{\prime}} = {\rm span}(\bm{V}_r)  \nonumber \\
		&\widetilde{\bm{\mathcal{U}}}_{r^{\prime}} = {\rm span}(\widetilde{\bm{U}}_{r^{\prime}} ), \quad \widetilde{\bm{\mathcal{V}}}_{r^{\prime}} = {\rm span}(\widetilde{\bm{V}}_{r^{\prime}} )
	\end{align*}
	and orthonormal matrices 
	\begin{align}\label{B_l and B_r}
		& \bm{B}_L := [\bm{U}_r \quad  \bm{U}^{\prime}_{1,r} \quad \bm{U}^{\prime}_{2,r^{\prime} - r} \quad \bm{U}^{\prime \prime}_{n - r - r^{\prime}}] \in \mathbb{R}^{n \times n} \nonumber \\
		& \bm{B}_R := [\bm{V}_r \quad  \bm{V}^{\prime}_{1,r} \quad \bm{V}^{\prime}_{2,r^{\prime} - r}\quad  \bm{V}^{\prime \prime}_{n - r - r^{\prime}}] \in \mathbb{R}^{n \times n}.
	\end{align}
	For definitions of the submatrices, see \cite{daei2018optimal}.
\end{lem}
Lemma \ref{Lem 3 } results in the following relation:
\begin{align}
	\widetilde{\bm{U}}_{r^{\prime}} = \bm{B}_L\begin{bmatrix}
		\cos \mathbf{\bm{\theta}}_{u} & \\
		-\sin \mathbf{\bm{\theta}}_{u} &  \\  
		& -\bm{I}_{r^{\prime} -r}\\
		& \bm{0} 
	\end{bmatrix}.
\end{align}
Then orthogonal projections onto the subspaces $ \widetilde{\bm{\mathcal{U}}}_{r^{\prime}} $ and $ \widetilde{\bm{\mathcal{U}}}_{r^{\prime}}^{\perp}  $ are:
\begin{align*}
	&\bm{P}_{\widetilde{\bm{\mathcal{U}}}_{r^{\prime}}} = \widetilde{\bm{U}}_{r^{\prime}}\widetilde{\bm{U}}_{r^{\prime}}^{\rm H}  \nonumber \\
	& = \bm{B}_L \small{\begin{bmatrix}
			\cos^2 \mathbf{\bm{\theta}}_{u} & -\sin \mathbf{\bm{\theta}}_{u} \cos \mathbf{\bm{\theta}}_{u} &  &  \\
			-\sin \mathbf{\bm{\theta}}_{u} \cos \mathbf{\bm{\theta}}_{u} & \sin^2 \mathbf{\bm{\theta}}_{u} &   &   \\
			& & \bm{I}_{r^{\prime} - r} & \\
			& &  & \bm{0} 
	\end{bmatrix}}
	\bm{B}_L^{\rm H},
\end{align*}
\normalfont
\begin{align*}
	& \bm{P}_{\widetilde{\bm{\mathcal{U}}}_{r^{\prime}}^{\perp}} = \bm{I} - \bm{P}_{\widetilde{\bm{\mathcal{U}}}_{r^{\prime}}}  \nonumber \\
	&  = \bm{B}_L \small{\begin{bmatrix}
			\sin^2 \mathbf{\bm{\theta}}_{u}  & \sin \mathbf{\bm{\theta}}_{u} \cos \mathbf{\bm{\theta}}_{u} &  & \\
			\sin \mathbf{\bm{\theta}}_{u} \cos \mathbf{\bm{\theta}}_{u} & \cos^2 \mathbf{\bm{\theta}}_{u} &  & \\
			&  & \bm{0}_{r^{\prime} - r} &   \\
			&  &  & \bm{I}_{ n - r^{\prime}- r}
	\end{bmatrix}}
	\bm{B}_L^{\rm H}.
\end{align*}
\normalfont	
Also, we have
\begin{align}\label{34}
	&\bm{Q}_{\widetilde{\bm{\mathcal{U}}}_{r^{\prime}}} := \widetilde{\bm{U}}_{r^{\prime}}\bm{\Lambda}\widetilde{\bm{U}}_{r^{\prime}}^{\rm{H}} + \bm{P}_{\widetilde{\bm{\mathcal{U}}}_{ r^{\prime}}^{\perp}}  \nonumber \\& 
	= \bm{B}_L
	\left[\small{\begin{array}{ccc}
			\bm{\Lambda}_1 \cos^2 \mathbf{\bm{\theta}}_{u} + \sin^2 \mathbf{\bm{\theta}}_{u}\\
			(\bm{I}-\bm{\Lambda}_1) \sin \mathbf{\bm{\theta}}_{u} \cos \mathbf{\bm{\theta}}_{u}\\
			\\ \\
	\end{array}}\right.\nonumber\\
	&\quad \quad \quad \quad \quad \quad \left.\small{\begin{array}{ccc}
			(\bm{I}-\bm{\Lambda}_1) \sin\mathbf{\bm{\theta}}_{u} \cos \mathbf{\bm{\theta}}_{u}&  & \\
			\bm{\Lambda}_1 \sin^2 \mathbf{\bm{\theta}}_{u} + \cos^2 \mathbf{\bm{\theta}}_{u} & & \\
			&\bm{\Lambda}_2&  \\
			&  &\bm{I}_{n-r^{\prime}-r}
	\end{array}}\right]\bm{B}_L^{\rm H},
\end{align}
\normalfont	
where  $ \bm{\Lambda} : = \begin{bmatrix}
	\bm{\Lambda}_1 \in \mathbb{R}^{r \times r}& \\
	& \bm{\Lambda}_2 \in \mathbb{R}^{r^{\prime}-r \times r^{\prime}-r}
\end{bmatrix} $.

We will rewrite $\bm{Q}_{\widetilde{\bm{\mathcal{U}}}_{r^{\prime}}}$ to incorporate an upper-triangular matrix. First, define the orthonormal base:
\begin{multline}
	\bm{O}_L := \left[\small{
		\begin{matrix}
			(\bm{\Lambda}_1 \cos^2 \mathbf{\bm{\theta}}_{u} + \sin^2 \mathbf{\bm{\theta}}_{u}).\bm{\Delta}_L^{-1}  \\
			-(\bm{I}-\bm{\Lambda}_1) \sin \mathbf{\bm{\theta}}_{u} \cos \mathbf{\bm{\theta}}_{u} .\bm{\Delta}_L^{-1}\\
			\\ 
			\\
	\end{matrix}}\right.               
	\\
	\left.
	\small{\begin{matrix}
			-(\bm{I}-\bm{\Lambda}_1) \sin \mathbf{\bm{\theta}}_{u} \cos \mathbf{\bm{\theta}}_{u} .\bm{\Delta}_L^{-1} &  &   \\
			(\bm{\Lambda}_1 \cos^2 \mathbf{\bm{\theta}}_{u} + \sin^2 \mathbf{\bm{\theta}}_{u}) .\bm{\Delta}_L^{-1}  &  &  \\ 
			& \bm{I}_{r^{\prime}-r} &   \\
			&  & \bm{I}_{n-r^{\prime}-r}
	\end{matrix}}\right],
\end{multline}
\normalfont
where $ \bm{\Delta}_L := \sqrt{\bm{\Lambda}_1^2 \cos^2 \mathbf{\bm{\theta}}_{u} + \sin^2 \mathbf{\bm{\theta}}_{u}} \in \mathbb{R}^{n \times n }$  is an invertible matrix since $ \bm{\Lambda}_1 \succeq \bm{0} $. We can rewrite \eqref{34} as
\begin{align}
	&\bm{Q}_{\widetilde{\bm{\mathcal{U}}}_{r^{\prime}}} =\bm{B}_L (\bm{O}_L\bm{O}_L^{\rm H})
	\left[\small{\begin{array}{ccc}
			\bm{\Lambda}_1 \cos^2 \mathbf{\bm{\theta}}_{u} + \sin^2 \mathbf{\bm{\theta}}_{u}\\
			(\bm{I}-\bm{\Lambda}_1) \sin \mathbf{\bm{\theta}}_{u} \cos \mathbf{\bm{\theta}}_{u}\\
			\\
			\\
	\end{array}}\right.\nonumber\\
	&\quad \quad \quad \left.\small{\begin{array}{ccc}
			(\bm{I}-\bm{\Lambda}_1) \sin\mathbf{\bm{\theta}}_{u} \cos \mathbf{\bm{\theta}}_{u}& & \\
			\bm{\Lambda}_1 \sin^2 \mathbf{\bm{\theta}}_{u} + \cos^2 \mathbf{\bm{\theta}}_{u} & & \\
			&\bm{\Lambda}_2 & \\
			& &\bm{I}_{n-r^{\prime}-r}
	\end{array}}\right]\bm{B}_L^{\rm H} \nonumber \\
	&=\small{\bm{B}_L \bm{O}_L \begin{bmatrix}
			\bm{\Delta}_L &(\bm{I}-\bm{\Lambda}_{1}^{2}) \sin \mathbf{\bm{\theta}}_{u} \cos \mathbf{\bm{\theta}}_{u}.\bm{\Delta}_L^{-1}& &\\
			&\bm{\Lambda}_1\bm{\Delta}_L^{-1} & &\\
			&&\bm{\Lambda}_2&  \\
			& & &\bm I
		\end{bmatrix}
				\bm{B}_L^{\rm H}} \nonumber \\
			& =: \bm{B}_L \bm{O}_L \begin{bmatrix}
				\bm{L}_{11} & \bm{L}_{12} & & \\
				& \bm{L}_{22} & & \\ 
				&‌&\bm{\Lambda}_2‌& \\
				& & &\bm{I}_{n-r^{\prime}-r}
			\end{bmatrix} \bm{B}_L^{\rm H}  \nonumber \\
			& = \bm{B}_L \bm{O}_L \bm{L} \bm{B}_L^{\rm H}, 
		\end{align}
		where $\bm{L} \in \mathbb{R}^{n \times n }$ is  a block upper-triangular matrix:
		
		\begin{align}\label{37}
			&\bm{L} :=\small{\begin{bmatrix}
					\bm{L}_{11} & \bm{L}_{12} & & \\
					& \bm{L}_{22} & & \\ 
					&‌&\bm{\Lambda}_2‌& \\
					& & &\bm{I}_{n-r^{\prime}-r}
			\end{bmatrix}} \nonumber \\
			&=\small{\begin{bmatrix}
					\bm{\Delta}_L  &  (\bm{I}-\bm{\Lambda}_{1}^{2}) \sin \mathbf{\bm{\theta}}_{u} \cos \mathbf{\bm{\theta}}_{u}.\bm{\Delta}_L^{-1}  & & \\
					& \bm{\Lambda}_1\bm{\Delta}_L^{-1} &  & \\
					& & \bm{\Lambda}_2& \\
					& & &\bm{I}_{n-r^{\prime}-r} \\
			\end{bmatrix}}.
		\end{align}
		
		\normalfont
		Since $\bm{B}_L$ and $\bm{O}_L$ are orthonormal bases, it follows that: 
		\begin	{align}\label{38}
		\| \bm{Q}_{\widetilde{\bm{\mathcal{U}}}_{r^{\prime}}} \| = \| \bm{L} \| = 1 .
	\end{align} 
	Similar results can  also be deduced for the row subspace:
	\begin{align}
		&\bm{R} :=\small{\begin{bmatrix}
				\bm{R}_{11} & \bm{R}_{12} & & \\
				& \bm{R}_{22} & & \\
				& &\bm{\Gamma}_2 & \\
				& & & \bm{I}_{n-r^{\prime}-r}
		\end{bmatrix} } \nonumber \\
		&= \small{\begin{bmatrix}
				\bm{\Delta}_R  &  (\bm{I}-\bm{\Gamma}_{1}^{2}) \sin \mathbf{\bm{\theta}}_{v} \cos \mathbf{\bm{\theta}}_{v}.\bm{\Delta}_R^{-1}  & &  \\
				& \bm{\Gamma}_{1}\bm{\Delta}_R^{-1} & & \\
				& & \bm{\Gamma}_2 & \\
				& & & \bm{I}_{n-r^{\prime}-r}  
		\end{bmatrix}},
	\end{align}
	\normalfont
	where $ \bm{\Delta}_R := \sqrt{\bm{\Gamma}_{1}^{2} \cos^2 \mathbf{\bm{\theta}}_{v} + \sin^2 \mathbf{\bm{\theta}}_{v}} $ has similar properties as $ \bm{\Delta}_L $. For an arbitrary matrix $ \bm{H} \in \mathbb{R}^{n \times n } $  we will have:
	\begin{align}\label{40}
		&\bm{Q}_{\widetilde{\bm{\mathcal{U}}}_{r^{\prime}}}\bm{H}\bm{Q}_{\widetilde{\bm{\mathcal{V}}}_{r^{\prime}}} = \bm{B}_L \bm{O}_L \bm{L} (\bm{B}_L^{\rm H} \bm{H} \bm{B}_R) \bm{R}^{\rm H} \bm{O}_R^{\rm H}  \bm{B}_R^{\rm H}   \nonumber \\
		&= \bm{B}_L \bm{O}_L \bm{L} \overline{\bm{H}} \bm{R}^{\rm H} \bm{O}_R^{\rm H} \bm{B}_R^{\rm H} \quad (\overline{\bm{H}}:=\bm{B}_L^{\rm H} \bm{H} \bm{B}_R )  \nonumber \\
		&=:\bm{B}_L \bm{O}_L \bm{L}
		\small{\begin{bmatrix}
				\overline{\bm{H}}_{11} & \overline{\bm{H}}_{12} &‌\overline{\bm{H}}_{13} & \overline{\bm{H}}_{14}\\
				\overline{\bm{H}}_{21} & \overline{\bm{H}}_{22} & \overline{\bm{H}}_{23} & \overline{\bm{H}}_{24}\\
				\overline{\bm{H}}_{31} & \overline{\bm{H}}_{32} &‌\overline{\bm{H}}_{33} & \overline{\bm{H}}_{34}\\
				\overline{\bm{H}}_{41} &  \overline{\bm{H}}_{42} & \overline{\bm{H}}_{43}& \overline{\bm{H}}_{44}
		\end{bmatrix}}
		\normalfont
		\bm{R}^{\rm H} \bm{O}_R^{\rm H} \bm{B}_R^{\rm H}.
	\end{align}
	Since ${\rm span}(\bm{X}_r) = {\rm span}(\bm{U}_r)$ and ${\rm span}(\bm{X}_r^{\rm H}) = {\rm span}(\bm{V}_r)$ and with upper triangular matrices $\bm{L}$ and $\bm{R}$, we can rewrite $ \bm{Q}_{\widetilde{\bm{\mathcal{U}}}_{r^{\prime}}}  \bm{X}_r \bm{Q}_{\widetilde{\bm{\mathcal{V}}}_{r^{\prime}}} $ in terms of new bases:
	\begin{align}\label{42}
		&\bm{Q}_{\widetilde{\bm{\mathcal{U}}}_{r^{\prime}}}  \bm{X}_r \bm{Q}_{\widetilde{\bm{\mathcal{V}}}_{r^{\prime}}} = \bm{B}_L \bm{O}_L \bm{L} (\bm{B}_L^{\rm H} \bm{X}_r \bm{B}_R) \bm{R}^{\rm H} \bm{O}_R^{\rm H}  \bm{B}_R^{\rm H}   \nonumber \\
		&= \bm{B}_L \bm{O}_L \bm{L} \overline{\bm{X}}_r \bm{R}^{\rm H} \bm{O}_R^{\rm H} \bm{B}_R^{\rm H} \quad (\overline{\bm{X}}_r:=\bm{B}_L^{\rm H} \bm{X}_r \bm{B}_R )  \nonumber \\
		&=:\bm{B}_L \bm{O}_L \bm{L}
		\begin{bmatrix} 
			\overline{\bm{X}}_{r,11} &  \\
			&  \bm{0}_{n-r} 
		\end{bmatrix}\bm{R}^{\rm H} \bm{O}_R^{\rm H} \bm{B}_R^{\rm H} \nonumber \\ 
		&=\bm{B}_L \bm{O}_L 
		\begin{bmatrix}
			\bm{L}_{11}\overline{\bm{X}}_{r,11}\bm{R}_{11} &  \\
			&  \bm{0}_{n-r} 
		\end{bmatrix}\bm{O}_R^{\rm H} \bm{B}_R^{\rm H}.
	\end{align}
	\begin{lem}\label{lem 4}
		The operator norms regarding the sub-blocks of L in \eqref{37} are as follows:
		\begin{align}
			&\|\bm{L}_{11}\| = \|\bm{\Delta}_L\| = \max_{i} \sqrt{ \lambda_{1}^{2}(i)\cos^2 \mathbf{\theta}_{u}(i) + \sin^2 \mathbf{\theta}_{u}(i)}, \nonumber \\ 
			& \|\bm{L}_{12}\| = \max_i \sqrt{\frac{(1-\lambda_{1}^{2}(i))^2 \cos^2 \mathbf{\theta}_{u}(i) \sin^2 u_i}{\lambda_{1}^{2}(i)\cos^2 \mathbf{\theta}_{u}(i) + \sin^2 \mathbf{\theta}_{u}(i)}},  \nonumber \\
			& \|\bm{I}_r - \bm{L}_{22} \| = \max_i \frac{\lambda_{1}(i)-\sqrt{\lambda_{1}^{2}(i)\cos^2 \mathbf{\theta}_{u}(i) + \sin^2 \mathbf{\theta}_{u}(i)}}{\sqrt{\lambda_{1}^{2}(i)\cos^2\mathbf{\theta}_{u}(i) + \sin^2 \mathbf{\theta}_{u}(i)}},  \nonumber \\
			& \label{eq:L11L12}\|[\bm{L}_{11} \quad  \bm{L}_{12}]\| = \max_i  \sqrt{  \frac{\lambda_{1}^{4}(i) \cos^2 \mathbf{\theta}_{u}(i) + \sin^2 \mathbf{\theta}_{u}(i)}{\lambda_{1}^{2}(i) \cos^2 \mathbf{\theta}_{u}(i) + \sin^2 \mathbf{\theta}_{u}(i) } } \\ 
			& \label{eq:L'} \|\bm L ^ \prime \|^{2} = \max_i {d}_i(\bm{\theta}_{u},\bm \lambda_{1},\bm \lambda_{2}) \\
			& \Big\|\begin{bmatrix}
				\bm{I}_{r} - \bm{L}_{22} & \\
				& \bm{I}_{r} - \bm{\Lambda}_{2}
			\end{bmatrix} \Big\| =  \max ~ \Big \{  \max_i (1-\lambda_{2}(i))\nonumber \\
			&‌ \quad \quad \quad \quad \quad \quad , \max_i \Big(1 - \frac{\lambda_{1}(i)}{\sqrt{\lambda_{1}^{2}(i)\cos^2 \mathbf{\theta}_{u}(i) + \sin^2 \mathbf{\theta}_{u}(i)}}\Big) \Big \}, 
		\end{align} 
		where $d_i$ is defined in \eqref{eq:d1} and \eqref{eq:d2}. The same  equalities hold for sub-blocks of $ \bm{R} $. 
		
		proof.  See Appendix \ref{proof lemma 4}.
	\end{lem}
	\subsection{Support Definitions}
	Suppose that $\bm{X}_r \in \mathbb{R}^{n \times n}$ is a rank-$r$ matrix obtained via the truncated SVD of $ \bm{X} $: 
	$$ \bm{X} = \bm{X}_r + \bm{X}_{r^+} = \bm{U}_r\overline{\bm{X}}_{r,11}\bm{V}_r^{\rm H} +‌ \bm{X}_{r^+},$$
	where $ \bm{U}_r $ and $ \bm{V}_r $ are some orthogonal bases of column and row spaces of $\bm{X}_r$, and therefore $ \overline{\bm{X}}_{r,11} $ is not necessarily diagonal. Also consider that  $ \bm{\mathcal{U}}_{r} = {\rm span}(\bm{U}_r)= {\rm span}(\bm{X}_r) $ and $ \bm{\mathcal{V}}_{r} = {\rm span}(\bm{V}_r)= {\rm span}(\bm{X}_r^{\rm H}) $ are column and row subspaces of $  \bm{X}_r$, respectively. Then the we define the support of $ \bm{X}_r $ by:  
	\begin{align}\label{44}
		&\bm{T} := \{\bm{Z}\in\mathbb{R}^{n\times n} : \bm{Z}=\bm{P}_{\bm{\mathcal{U}}_r}\bm{Z}\bm{P}_{\bm{\mathcal{V}}_r}+\bm{P}_{\bm{\mathcal{U}}_r}\bm{Z}\bm{P}_{\bm{\mathcal{V}}_r^\perp} \nonumber \\
		&\quad \quad \quad \quad \quad \quad 	+\bm{P}_{\bm{\mathcal{U}}_r^\perp}\bm{Z}\bm{P}_{\bm{\mathcal{V}}_r^\perp}\}   = {\rm supp}(\bm{X}_r),
	\end{align}
	and the orthogonal projection onto $\bm{T}$ and $\bm{T}^{\perp}$ as
	\begin{align}
		&\mathcal{P}_{\bm{T}}(\bm{Z}) = \bm{P}_{\bm{\mathcal{U}}}\bm{Z} +‌\bm{Z}\bm{P}_{\bm{\mathcal{V}}} - \bm{P}_{\bm{\mathcal{U}}}\bm{Z}\bm{P}_{\bm{\mathcal{V}}},  \mathcal{P}_{\bm{T}^{\perp}}(\bm{Z}) =  \bm{P}_{\bm{\mathcal{U}}^{\perp}}\bm{Z}\bm{P}_{\bm{\mathcal{V}}^{\perp }}.
	\end{align} 
	We can rewrite $ \bm{T} $  using Lemma \ref{Lem 3 } as 
	\begin{align}
		&\bm{T} = \Big\{\bm{Z}\in\mathbb{R}^{ n\times n} : \bm{Z}=\bm{B}_L \overline{\bm{Z}} \bm{B}_R^{\rm H}, \quad \overline{\bm{Z}}:=
		\begin{bmatrix}
			\overline{\bm{Z}}_{11} & \overline{\bm{Z}}_{12} \\
			\overline{\bm{Z}}_{21} & \bm{0}_{n-r}
		\end{bmatrix} \Big\}  \nonumber \\
		&=\bm{B}_L \overline{\bm{T}} \bm{B}_R^{\rm H},
	\end{align}
	where $ \overline{\bm{T}} \subset \mathbb{R}^{n \times n} $  is the support of $ \overline{\bm{X}}_r = \bm{B}_L^{\rm H} \bm{X}_r \bm{B}_R $:
	\begin{align}
		\overline{\bm{T}} = \{\overline{\bm{Z}}\in\mathbb{R}^{ m\times n} : \overline{\bm{Z}}:=	\begin{bmatrix}
			\overline{\bm{Z}}_{11} & \overline{\bm{Z}}_{12} \\
			\overline{\bm{Z}}_{21} & \bm{0}_{n-r}
		\end{bmatrix}\}.
	\end{align}
	For arbitrary 
	\begin{align}
		\overline{\bm{Z}}:=	\begin{bmatrix}
			\overline{\bm{Z}}_{11} & \overline{\bm{Z}}_{12} \\
			\overline{\bm{Z}}_{21} & \overline{\bm{Z}}_{22}
		\end{bmatrix} \in \mathbb{R}^{n\times n},
	\end{align}
	the orthogonal projection onto $ \overline{\bm{T}} $ and its complement $ \overline{\bm{T}}^{\perp} $ are  
	\begin{align}\label{47}
		&\mathcal{P}_{\overline{\bm{T}}}(\overline{\bm{Z}}) = \begin{bmatrix}
			\overline{\bm{Z}}_{11} & \overline{\bm{Z}}_{12} \\
			\overline{\bm{Z}}_{21} & \bm{0}_{n-r}
		\end{bmatrix},  \nonumber \\
		&\mathcal{P}_{\overline{\bm{T}}^{\perp}} (\overline{\bm{Z}})= \begin{bmatrix}
			\bm{0}_{r} & \\
			& \overline{\bm{Z}}_{22}
		\end{bmatrix}, 
	\end{align}
	respectively. When $ \bm{Z} = \bm{B}_L \overline{\bm{Z}} \bm{B}_R^{\rm H}  $, it follows that:   
	\begin{align}\label{48}
		& \mathcal{P}_{\bm{T}}(\bm{Z}) = \bm{B}_L \mathcal{P}_{\overline{\bm{T}}}(\overline{\bm{Z}})\bm{B}_R^{H}, \nonumber \\
		& \mathcal{P}_{\bm{T}^{\perp}}(\bm{Z}) = \bm{B}_L \mathcal{P}_{\overline{\bm{T}}^{\perp}}(\overline{\bm{Z}}) \bm{B}_R^{H}.
	\end{align}
	
	\section{Proof of Theorem \ref{our theorem}}
	\label{proof theorem 2}
	In Appendix \ref{Nullspace Property} we establish the null space property. Suppose $ \bm{H}:= \widehat{\bm{X}} - \bm{X} $ is  the error of solution to problem \eqref{ourproblem}. For the SVD decomposition $ \mathcal{ P}_{ \bm{T}^{ \perp}} (\bm{H}) = \hat{\bm{U}} \hat{\bm{\Sigma}} \hat{\bm{V}}^{\rm H}$, use the following partitions: 
	\begin{align}
		&\bm{\Sigma}_i = \hat{\bm{\Sigma}}[(i-1)\hat{r}+1 : i \hat{r} , (i-1)\hat{r}+1 : i \hat{r}] \in \mathbb{R}^{\hat{r} \times \hat{r}}  \nonumber \\
		& \bm{U}_i=\hat{\bm{U}}[:,(i-1)\hat{r}+1 : i \hat{r}] \in \mathbb{R}^{n \times \hat{r} }  \nonumber \\
		& \bm{V}_i=\hat{\bm{V}}[:,(i-1)\hat{r}+1 : i \hat{r}] \in \mathbb{R}^{n \times \hat{r} } \nonumber \\
		& \bm{H}_i := \bm{U}_i \bm{\Sigma}_i \bm{V}_i^{\rm H},
	\end{align}
	Decompose $\bm{H}$ as:
	\begin{align}\label{63} 
		\bm{H} = ‌\mathcal{P}_{\bm{T}^{\perp}}(\bm{H})  + \mathcal{P}_{\bm{T}}(\bm{H}) = \sum_{i \ge \bm{0} }\bm{H}_i,
	\end{align} 
	where $ \mathcal{P}_{\bm{T}}(\bm{H}) = \bm{H}_0$. Also, row and column spans of $\bm{H}_i$ and $ \bm{H}_j $ for $ i\neq j $ are orthogonal, in other words 
	\begin{align}
		\label{eq:orth}
		\bm{H}_i^{\rm H} \bm{H}_j = \bm{H}_j \bm{H}_i = \bm{0}_n,  \quad \quad i \neq j.
	\end{align} 
	Since $\bm{X}$ and $\widehat{\bm{X}}$ are in the feasible set of program \eqref{ourproblem}:
	\begin{align}\label{65}
		\| \mathcal{A}(\bm{H}) \|_{2} \le \| \mathcal{A}(\widehat{\bm{X}}) - \bm{y} \|_{2} + \| \mathcal{A}(\bm{X}) - \bm{y}  \|_{2} \le 2e.
	\end{align}
	Using \eqref{63}, \eqref{65} and triangle inequality, we will have:
	\begin{align}\label{66}
		&\| \mathcal{A}(\bm{H}_0 +\bm{H}_1) \|_{2} \le \sum_{i \ge 2 }\| \mathcal{A}(\bm{H}_i) \|_{2} + 2e.
	\end{align}
	Suppose that $\mathcal{A}$ satisfies \eqref{eq.RIP_A} with constant $\delta_{\tilde{r}}(\mathcal{A})$ for $ \tilde{r} \ge 2r + \hat{r} $: 
	\begin{align*}
		&{\rm rank }(\bm{H}_0 +\bm{H}_1) = {\rm rank }(\mathcal{P}_{\bm{T}}(\bm{H}) +\bm{H}_1) \le 2r + \hat{r} \le \tilde{r}  \nonumber \\ 
		&{\rm rank }(\bm{H}_i) \le \hat{r} \le \tilde{r} \quad : \quad i \ge 1.
	\end{align*}
	Combining \eqref{66} and \eqref{eq.RIP_A} yields
	\begin{align}\label{67}
		&(1-\delta_{\tilde{r}}(\mathcal{A}))\|\bm{H}_0 +\bm{H}_1 \|_{F} \le (1+\delta_{\tilde{r}}(\mathcal{A})) \sum_{i \ge 2 } \| \bm{H}_i\|_{F} +‌2e\nonumber \\ 
		&\le \frac{1+\delta_{\tilde{r}}(\mathcal{A})}{\sqrt{\hat{r}}} \sum_{i \ge 1 } \| \bm{H}_i\|_{*} +‌2e\stackrel{\eqref{eq:orth}}{=}  \frac{1+\delta_{ \tilde{r}}(\mathcal{A})}{\sqrt{\hat{r}}} \| \sum_{i \ge 1 }  \bm{H}_i\|_{*} + 2e \nonumber \\ 
		& =  \frac{1+\delta_{\tilde{r}}(\mathcal{A})}{\sqrt{\hat{r}}} \| \mathcal{P}_{\bm{T}^{\perp}}(\bm{H}) \|_{*} + 2e,
	\end{align} 
	Using null space property in \eqref{62}, we find that 
	\begin{align}
		\label{eq:h0+h1}
		&\|\bm{H}_0 +\bm{H}_1 \|_{F} \le \frac{1+\delta_{\tilde{r}}(\mathcal{A})}{(1-\delta_{\tilde{r}}(\mathcal{A})) \sqrt{\hat{r}}}  \| \mathcal{P}_{\bm{T}^{\perp}}(\bm{H}) \|_{*} + \frac{2e}{ 1-\delta_{\tilde{r}}(\mathcal{A})} \nonumber \\  
		& \le \frac{1+\delta_{\tilde{r}}(\mathcal{A})}{(1-\delta_{\tilde{r}}(\mathcal{A})) \sqrt{\hat{r}} }  (\alpha_3 \|\mathcal{P}_{\bm{T}}(\bm{H})\|_{*} + \alpha_4 \|\mathcal{P}_{\widetilde{\bm{T}}}(\bm{H})\|_{*} +‌2 \|\bm{X}_{r^{+}}\|_{*}) \nonumber \\ 
		&\quad \quad \quad \quad \quad \quad    +  \frac{2e}{1-\delta_{\tilde{r}}(\mathcal{A})}.
	\end{align}
	Since $ \mathcal{P}_{\bm{T}}(\bm{H}) = \bm H_0$ and $ {\rm rank}(\bm H_0) \le 2r$ by \eqref{44}, we will have $\|\bm H_0\|_* \leq \sqrt{2r}  \|\bm H_0\|_F$. The same is true for $ \mathcal{P}_{\widetilde{\bm{T}}}(\bm{H})$ by \eqref{60}. Next, $\bm H_1$ contains the $\hat{r}$ most powerful modes of $\mathcal{P}_{\bm{T}^{\perp}}(\bm{H})$ while for $ \widetilde{\bm{T}} \subset \bm{T}^{\perp}$, the projection  $ \mathcal{ P}_{ \widetilde{ \bm{T}}}(\bm{H})$ contains $2r \leqslant \hat{r}$ of its modes. Therefore, $\|\mathcal{P}_{ \widetilde{ \bm{T}}}(\bm{H})\|_F \leqslant \| \bm{H}_1\|_{F}$. Cauchy-Schwarz inequality can be used to obtain $\alpha_3 \| \bm{H}_0\|_{F} + \alpha_4 \| \bm{H}_1\|_{F} \leq \sqrt{\alpha_3^2 + \alpha_4^2}.\sqrt{\| \bm{H}_0\|_{F}^2 + \| \bm{H}_1\|_{F}^2}$. The orthogonality of $\bm H_i$ matrices in \eqref{eq:orth} allows us to write $ \| \bm{H}_0\|_F^2 + \| \bm{H}_1\|_F^2 = \| \bm{H}_0+ \bm{H}_1\|_F^2 $. Overall, \eqref{eq:h0+h1} will become:
	\begin{align}
		\label{68}
		\| \bm{H}_0+ \bm{H}_1\|_{F} & \leq  \frac{ 1+\delta_{\tilde{r}}(\mathcal{A})}{ 1-\delta_{\tilde{r}}(\mathcal{A})} \frac{\sqrt{2r}}{\sqrt{\hat{r}}} \sqrt{\alpha_3^2 + \alpha_4^2} \| \bm{H}_0+ \bm{H}_1\|_{F} \nonumber \\
		& + \frac{1+\delta_{\tilde{r}}(\mathcal{A})}{1-\delta_{\tilde{r}}(\mathcal{A})} \frac{2}{\sqrt{\hat{r}}} \|\bm{X}_{r^{+}}\|_{*}   +  \frac{2e}{1-\delta_{\tilde{r}}(\mathcal{A})},
	\end{align} 
	which can be written as:        
	\begin{align}
		\label{eq:h0+h1_frac}
		\|\bm{H}_0 +\bm{H}_1 \|_{F} \le  \frac{\frac{ 1+ \delta_{ \tilde{r}}(\mathcal{A})}{ 1- \delta_{ \tilde{r}}(\mathcal{A})} \frac{2}{\sqrt{\hat{r}}} \|\bm{X}_{r^{+}}\|_{*} + \frac{2e}{ 1-\delta_{\tilde{r}} }(\mathcal{A})}{ 1-\frac{1+ \delta_{ \tilde{r}}(\mathcal{A})}{1-\delta_{\tilde{r}}(\mathcal{A})} \sqrt{ \frac{2r}{\hat{r}} ( \alpha_3^2 + \alpha_4^2 )}	},
	\end{align} 
	when the denominator is positive which requires:
	\begin{align}\label{70}
		\delta_{\tilde{r}}(\mathcal{A}) \le	\frac{1- \sqrt{ \frac{2r}{\hat{r}} ( \alpha_3^2 + \alpha_4^2 )} }{ 1 + \sqrt{ \frac{2r}{\hat{r}} ( \alpha_3^2 + \alpha_4^2 )}}.
	\end{align}
	We return to the proof of the theorem by noting that:
	\begin{align}\label{72}
		&\| \sum_{i \ge 2 } \bm{H}_i \|_{F} \le \sum_{i \ge 2 } \| \bm{H}_i \|_{F} \le \frac{\|\mathcal{P}_{\bm{T}^{\perp}}(\bm{H}) \|_{*}}{\sqrt{\hat{r}}} \nonumber \\
		&\le \sqrt{ \frac{2r}{\hat{r}} ( \alpha_3^2 + \alpha_4^2 )} \| \bm{H}_0+ \bm{H}_1\|_{F} + ‌2 \frac{1}{\sqrt{\hat{r}}} \|\bm{X}_{r^{+}}\|_{*}  \nonumber \\ 
		& \stackrel{\eqref{eq:h0+h1_frac}}{\le}  \frac{ \frac{2}{\sqrt{\hat{r}}} \|\bm{X}_{r^{+}}\|_{*} + \frac{2e}{ 1-\delta_{\tilde{r}} (\mathcal{A})} \sqrt{ \frac{2r}{\hat{r}} ( \alpha_3^2 + \alpha_4^2 )} }{ 1-\frac{1+ \delta_{ \tilde{r}}(\mathcal{A})}{1-\delta_{\tilde{r}}(\mathcal{A})} \sqrt{ \frac{2r}{\hat{r}} ( \alpha_3^2 + \alpha_4^2 )}	}.
	\end{align} 
	Finally, \eqref{68} and \eqref{72} yield:    
	\begin{align}
		&\|\widehat{\bm{X}} - \bm{X} \|_{F} = \|\bm{H} \|_{F} \le \| \bm{H}_0+ \bm{H}_1\|_{F} + \| \sum_{i \ge 2 } \bm{H}_i \|_{F} \nonumber \\ 
		& \le \frac{ \frac{ 4 }{ ( 1- \delta_{ \tilde{r}}(\mathcal{A})) \sqrt{ \hat{r}}} \|\bm{X}_{r^{+}}\|_{*} + \frac{2e}{ 1-\delta_{\tilde{r}}(\mathcal{A})} \left( 1 + \sqrt{ \frac{2r}{\hat{r}} ( \alpha_3^2 + \alpha_4^2 )} \right) }{ 1-\frac{1+ \delta_{ \tilde{r}}(\mathcal{A})}{1-\delta_{\tilde{r}}(\mathcal{A})} \sqrt{ \frac{2r}{\hat{r}} ( \alpha_3^2 + \alpha_4^2 )}	},
	\end{align}
	when $\tilde{r} \ge 2r + \hat{r}$ and \eqref{70} is met. Taking $\hat{r} = 30r$ and $ \tilde{r} = 32r$ the proof of Theorem \ref{our theorem} is complete.

	\section{Null Space Property}
	\label{Nullspace Property}
	Suppose that $ \widehat{\bm{X}} $ and $ \bm{H}:= \widehat{\bm{X}} - \bm{X} $ are  the solution and its error of problem \eqref{ourproblem}. We have  
	\begin{align}\label{49}
		\| \bm{Q}_{\widetilde{\bm{\mathcal{U}}}_{r^{\prime}}} (\bm{X}+\bm{H}) \bm{Q}_{\widetilde{\bm{\mathcal{V}}}_{r^{\prime}}} \|_{*} \le  
		\| \bm{Q}_{\widetilde{\bm{\mathcal{U}}}_{r^{\prime}}} \bm{X} \bm{Q}_{\widetilde{\bm{\mathcal{V}}}_{r^{\prime}}} \|_{*}.
	\end{align}
	The R.H.S. is bounded as follows:
	\begin{align} \label{50}
		&\| \bm{Q}_{\widetilde{\bm{\mathcal{U}}}_{r^{\prime}}} \bm{X} \bm{Q}_{\widetilde{\bm{\mathcal{V}}}_{r^{\prime}}} \|_{*} \le \| \bm{Q}_{\widetilde{\bm{\mathcal{U}}}_{r^{\prime}}}  \bm{X}_r \bm{Q}_{\widetilde{\bm{\mathcal{V}}}_{r^{\prime}}} \|_{*} + \| \bm{Q}_{\widetilde{\bm{\mathcal{U}}}_{r^{\prime}}}  \bm{X}_{r^{+}} \bm{Q}_{\widetilde{\bm{\mathcal{V}}}_{r^{\prime}}} \|_{*} \nonumber \\
		& \stackrel{\eqref{42}}{=} \|  \bm{B}_L \bm{O}_L \begin{bmatrix}
			\bm{L}_{11}\overline{\bm{X}}_{r,11}\bm{R}_{11} & \\
			&‌\!\!\! \!\!\! \bm{0}
		\end{bmatrix} \bm{O}_R^{\rm H} \bm{B}_R^{\rm H}\|_{*} + \| \bm{Q}_{\widetilde{\bm{\mathcal{U}}}_{r^{\prime}}} \bm{X}_{r^{+}} \bm{Q}_{\widetilde{\bm{\mathcal{V}}}_{r^{\prime}}} \|_{*}   \nonumber \\
		&= \Bigg\|
		\begin{bmatrix}
			\bm{L}_{11}\overline{\bm{X}}_{r,11}\bm{R}_{11} & \\
			&‌\bm{0}
		\end{bmatrix}
		\Bigg\| _{*}  + \|\bm{Q}_{\widetilde{\bm{\mathcal{U}}}_{r^{\prime}}}  \bm{X}_{r^{+}} \bm{Q}_{\widetilde{\bm{\mathcal{V}}}_{r^{\prime}}} \|_{*}  \nonumber \\
		& =\|\bm{L}\overline{\bm{X}}\bm{R}^{\rm H} \|_{*}‌ + \| \bm{Q}_{\widetilde{\bm{\mathcal{U}}}_{r^{\prime}}} \bm{X}_{r^{+}} \bm{Q}_{\widetilde{\bm{\mathcal{V}}}_{r^{\prime}}} \|_{*},
	\end{align}
	The L.H.S. of the \eqref{49} can be bounded as:
	\begin{align}
		& \| \bm{Q}_{\widetilde{\bm{\mathcal{U}}}_{r^{\prime}}} ( \bm{X}+\bm{H}) \bm{Q}_{\widetilde{\bm{\mathcal{V}}}_{r^{\prime}}} \|_{*}  \nonumber \\
		& \ge  \| \bm{Q}_{\widetilde{\bm{\mathcal{U}}}_{r^{\prime}}} ( \bm{X}_r+\bm{H}) \bm{Q}_{\widetilde{\bm{\mathcal{V}}}_{r^{\prime}}} \|_{*} -\| \bm{Q}_{\widetilde{\bm{\mathcal{U}}}_{r^{\prime}}}  \bm{X}_{r^{+}} \bm{Q}_{\widetilde{\bm{\mathcal{V}}}_{r^{\prime}}} \|_{*}    \nonumber \\
		& \stackrel{\eqref{40}}{=}   \|\bm{L} \overline{ \bm{X}}_r \bm{R}^{ \rm H} + \bm{L}\overline{\bm{H}}\bm{R}^{\rm H}\|_{*} - \| \bm{Q}_{\widetilde{\bm{\mathcal{U}}}_{r^{\prime}}}  \bm{X}_{r^{+}} \bm{Q}_{\widetilde{\bm{\mathcal{V}}}_{r^{\prime}}} \|_{*}  \nonumber \\
		& = \|\bm{L} (\overline{\bm{X}}_r + \mathcal{P}_{ \overline{ \bm{T}}}( \overline{\bm{H}}) + \mathcal{P}_{ \overline{ \bm{T}}^{ \perp}}(\overline{\bm{H}}) ) \bm{R}^{\rm H} \nonumber \\ 
		& \quad \quad \quad -\mathcal{P}_{\overline{\bm{T}}^{\perp}}(\overline{\bm{H}})+ \mathcal{P}_{\overline{\bm{T}}^{\perp}}(\overline{\bm{H}}) ‌\|_{*} -  \| \bm{Q}_{\widetilde{\bm{\mathcal{U}}}_{r^{\prime}}}  \bm{X}_{r^{+}} \bm{Q}_{\widetilde{\bm{\mathcal{V}}}_{r^{\prime}}} \|_{*}. 
	\end{align}
	Using the definitions in \eqref{47}, we will have:   
	\begin{align}
		&\| \bm{Q}_{\widetilde{\bm{\mathcal{U}}}_{r^{\prime}}} ( \bm{X}+\bm{H}) \bm{Q}_{\widetilde{\bm{\mathcal{V}}}_{r^{\prime}}} \|_{*} \ge  \|\bm{L} ( \overline{\bm{X}}_r  + \mathcal{ P}_{ \overline{ \bm{ T}}} ( \overline{ \bm{H}}) + \overline{\bm{H}}^{\prime \prime}) \bm{R}^{\rm H} -  \nonumber \\
		& \overline{\bm{H}}^{\prime \prime} + \mathcal{P}_{ \overline{ \bm{T}}^{ \perp}}(\overline{\bm{H}})‌\|_{*} - \| \bm{Q}_{\widetilde{\bm{\mathcal{U}}}_{r^{\prime}}}  \bm{X}_{r^{+}} \bm{Q}_{\widetilde{\bm{\mathcal{V}}}_{r^{\prime}}}  \|_{*} = \|\bm{L} ( \overline{\bm{X}}_r + \mathcal{P}_{ \overline{ \bm{T}}}( \overline{ \bm{H}}) + \nonumber \\ 
		&  \overline{\bm{H}}^{\prime} ) \bm{R}^{\rm H}-\overline{\bm{H}}^{ \prime} + \mathcal{P}_{\overline{\bm{T}}^{\perp}}(\overline{\bm{H}})‌\|_{*} - ‌\| \bm{Q}_{\widetilde{\bm{\mathcal{U}}}_{r^{\prime}}}  \bm{X}_{r^{+}} \bm{Q}_{\widetilde{\bm{\mathcal{V}}}_{r^{\prime}}}  \|_{*} \ge  \nonumber \\ 
		& \|\bm{L} ( \overline{\bm{X}}_r + \mathcal{P}_{ \overline{\bm{T}}}( \overline{\bm{H}}) ) \bm{R}^{\rm H} + \mathcal{P}_{\overline{\bm{T}}^{\perp}}(\overline{\bm{H}}) \|_{*} - \|\bm{L} \mathcal{P}_{ \overline{ \bm{T}}}( \overline{ \bm{H}}) \bm{R}^{\rm H} \|_{*} - \nonumber \\ 
		& \| \bm{L}\overline{\bm{H}}^{\prime} \bm{R}^{\rm H}-\overline{\bm{H}}^{\prime} \|_{*} - \|\bm{Q}_{\widetilde{\bm{\mathcal{U}}}_{r^{\prime}}}  \bm{X}_{r^{+}} \bm{Q}_{\widetilde{\bm{\mathcal{V}}}_{r^{\prime}}}  \|_{*},
	\end{align} 
	where we have defined:
	\small{\begin{align} \label{51}
			\overline{\bm{H}}^{\prime} \!\!\!\! := \!\!\! \begin{bmatrix}
				\bm{0} \!\!\!\!&  & & \\
				& \overline{\bm{H}}_{22} & \overline{\bm{H}}_{23} & \overline{\bm{H}}_{23} \\
				& \overline{\bm{H}}_{32} &‌\overline{\bm{H}}_{33} & \overline{\bm{H}}_{34} \\
				& \overline{\bm{H}}_{42} & \overline{\bm{H}}_{43} & 
				\bm{0}
			\end{bmatrix} , \, 
			\overline{\bm{H}}^{\prime \prime} \!\!\! := \!\!\! \begin{bmatrix}
				\bm{0} \!\!\!\!&  & & \\
				& \overline{\bm{H}}_{22} & \overline{\bm{H}}_{23} & \overline{\bm{H}}_{23} \\
				& \overline{\bm{H}}_{32} &‌\overline{\bm{H}}_{33} & \overline{\bm{H}}_{34} \\
				& \overline{\bm{H}}_{42} & \overline{\bm{H}}_{43} & 
				\overline{\bm{H}}_{44} 
			\end{bmatrix}.
	\end{align}}
	\normalfont
	Consequently, we have:
	\begin{align} \label{52}
		&\| \bm{Q}_{\widetilde{\bm{\mathcal{U}}}_{r^{\prime}}} ( \bm{X}+\bm{H}) \bm{Q}_{\widetilde{\bm{\mathcal{V}}}_{r^{\prime}}} \|_{*} \stackrel{\eqref{42}}{\ge} \| \textrm {diag} [ \bm{L}_{11} \overline{\bm{X}}_{ r,11} \bm{R}_{ 11} \; , \; \bm{0}_{n-r} ]  +‌ \nonumber \\
		& \overline{ \bm{H}}^{\prime \prime} \|_* - \|\bm{L}\mathcal{P}_{\overline{\bm{T}}}(\overline{\bm{H}})\bm{R}^{\rm H} \|_{*} - \| \bm{L}\overline{\bm{H}}^{\prime} \bm{R}^{\rm H}-\overline{\bm{H}}^{ \prime } \|_{*} -  \nonumber \\
		&  \| \bm{Q}_{\widetilde{\bm{\mathcal{U}}}_{r^{\prime}}}  \bm{X}_{r^{+}} \bm{Q}_{\widetilde{\bm{\mathcal{V}}}_{r^{\prime}}}  \|_{*} = \| \textrm {diag} [ \bm{L}_{11} \overline{\bm{X}}_{ r,11} \bm{R}_{ 11} \; , \; \bm{0} ]\|_* +‌ \| \overline{\bm{H}}^{\prime \prime} \|_* - \nonumber \\
		& \| \bm{L} \mathcal{P}_{ \overline{ \bm{T}}}( \overline{ \bm{H}}) \bm{R}^{\rm H} \|_{*} - \| \bm{L}\overline{\bm{H}}^{ \prime }\bm{R}^{\rm H}-\overline{\bm{H}}^{ \prime } \|_{*}  - \|\bm{Q}_{\widetilde{\bm{\mathcal{U}}}_{r^{\prime}}}  \bm{X}_{r^{+}} \bm{Q}_{\widetilde{\bm{\mathcal{V}}}_{r^{\prime}}} \|_{*}\nonumber \\
		& \stackrel{\eqref{47}}{=} \|\bm{L}\overline{\bm{X}}\bm{R}^{\rm H} \|_{*}‌  +‌\| \mathcal{P}_{\overline{\bm{T}}^{\perp}}(\overline{\bm{H}}) \|_{*} -\| \bm{L}\mathcal{P}_{\overline{\bm{T}}}(\overline{\bm{H}})\bm{R}^{\rm H}\|_{*}   \nonumber \\
		& -\| \bm{L}\overline{\bm{H}}^{\prime }\bm{R}^{\rm H}-\overline{\bm{H}}^{\prime } \|_{*} - \| \bm{Q}_{\widetilde{\bm{\mathcal{U}}}_{r^{\prime}}}  \bm{X}_{r^{+}} \bm{Q}_{\widetilde{\bm{\mathcal{V}}}_{r^{\prime}}} \|_{*},
	\end{align} 
	where we have used the fact that $ \| \bm{A} + \bm{B}\|_{*} = \| \bm{A}\|_{*} + \|\bm{B} \|_{*}$ when the column and row spaces of $ \bm{A} $ are orthogonal to $ \bm{B} $. Combining \eqref{49} with the upper and lower bounds in \eqref{50} and \eqref{52} yields:
	\begin{align}\label{53}
		&‌\| \mathcal{P}_{\overline{\bm{T}}^{\perp}}(\overline{\bm{H}}) \|_{*} \le \| \bm{L}\mathcal{P}_{\overline{\bm{T}}}(\overline{\bm{H}})\bm{R}^{\rm H}\|_{*} + \| \bm{L}\overline{\bm{H}}^{\prime} \bm{R}^{\rm H}-\overline{\bm{H}}^{\prime} \|_{*} \nonumber \\
		&\quad \quad \quad \quad \quad \quad \quad \quad \quad \quad \quad + 2\| \bm{Q}_{\widetilde{\bm{\mathcal{U}}}_{r^{\prime}}}  \bm{X}_{r^{+}} \bm{Q}_{\widetilde{\bm{\mathcal{V}}}_{r^{\prime}}}  \|_{*}.
	\end{align}
	Note that
	\begin{align}
		&\textrm {diag} [ \bm{0}_r ,  \bm{L}_{22} , \bm{\Lambda}_2 , \bm{I} ]‌\mathcal{ P}_{ \overline{ \bm{T}}}( \overline{\bm{H}})  \textrm {diag} [ \bm{0}_r ,  \bm{R}_{22} , \bm{\Gamma}_2 , \bm{I} ] \stackrel{\eqref{47}}{=} \nonumber \\ 
		& \textrm {diag} [ \bm{0}_r , \bm{L}_{22} , \bm{\Lambda}_2 , \bm{I} ] 
		\small{\begin{bmatrix}
				\overline{\bm{H}}_{11} & \overline{\bm{H}}_{12} &‌\overline{\bm{H}}_{13} &  ‌\overline{\bm{H}}_{14}\\
				\overline{\bm{H}}_{21} &   \\
				\overline{\bm{H}}_{31} &   & \bm{0}_{n-r}&  \\
				‌\overline{\bm{H}}_{41} &
		\end{bmatrix} }
		\normalfont \nonumber\\ 
		&  \quad \quad \quad  \textrm {diag} [ \bm{0}_r ,  \bm{R}_{22} , \bm{\Gamma}_2 , \bm{I} ]   = \bm{0}_n.
	\end{align}
	In the R.H.S. of \eqref{53}, we will have:
	\begin{align}
		\label{55}
		& \| \bm{L} \mathcal{P}_{ \overline{ \bm{T}}}( \overline{ \bm{H}}) \bm{R}^{\rm H} \|_{*}  = \| \bm{ L} \mathcal{P}_{ \overline{ \bm{T}}}( \overline{ \bm{H}}) \bm{R}^{\rm H} - \nonumber \\  
		&  \textrm {diag} [ \bm{0}_r ,  \bm{L}_{22} , \bm{\Lambda}_2 , \bm{I} ]‌\mathcal{ P}_{ \overline{ \bm{T}}}( \overline{\bm{H}})  \textrm {diag} [ \bm{0}_r ,  \bm{R}_{22} , \bm{\Gamma}_2 , \bm{I} ] \|_* \stackrel{\eqref{56}}{=} \nonumber \\
		&\Bigg\|
		\begin{bmatrix}
			\bm{L}_{11} & \bm{L}_{12} & \\
			& \bm{0}_r & \\
			& &\bm{0} 
		\end{bmatrix}
		\mathcal{P}_{\overline{\bm{T}}}(\overline{\bm{H}}) \bm{R}^{\rm H}+
		\nonumber \\
		& \textrm {diag} [ \bm{0}_r ,  \bm{L}_{22} , \bm{\Lambda}_2 , \bm{I} ] \mathcal{P}_{\overline{\bm{T}}}(\overline{\bm{H}})
		\begin{bmatrix}
			\bm{R}_{11} & &\\
			\bm{R}^{\rm H}_{12}& \bm{0} & \\
			& &\bm{0} 
		\end{bmatrix} \Bigg\|_{*}
		\nonumber \\
		& \le \Bigg\|
		\begin{bmatrix}
			\bm{L}_{11} & \bm{L}_{12}  & \\
			& \bm{0}_r  & \\
			& &\bm{0} 
		\end{bmatrix}
		\mathcal{P}_{\overline{\bm{T}}}(\overline{\bm{H}}) \bm{R}^{\rm H}
		\Bigg\|_{*}   \nonumber \\
		& + \Bigg\|
		\textrm {diag} [ \bm 0_r , \bm L_{22} , \bm \Lambda_2 , \bm I ]
		\mathcal{P}_{\overline{\bm{T}}}(\overline{\bm{H}}) \begin{bmatrix}
			\bm{R}_{11} & & \\
			\bm{R}^{\rm H}_{12}& \bm{0}_r  & \\
			& &\bm{0}  
		\end{bmatrix}
		\Bigg\|_{*} \nonumber \\ 
		&\le \|[\bm{L}_{11} \quad  \bm{L}_{12}]\|\|\mathcal{P}_{\overline{\bm{T}}}(\overline{\bm{H}})\|_{*} \|\bm{R}\| +   \nonumber \\ 
		& \quad \quad \max [\|\bm{L}_{22}\|, \| \bm{\Lambda}_2 \| ,1] ‌\|\mathcal{P}_{ \overline{ \bm{T}}}(\overline{\bm{H}})\|_{*}\|[\bm{R}_{11} \quad  \bm{R}_{12}]\| \nonumber \\
		& \le ‌(\|[\bm{L}_{11} \quad  \bm{L}_{12}]\| +‌\|[\bm{R}_{11} \quad  \bm{R}_{12}]\|) \|\mathcal{P}_{\overline{\bm{T}}}(\overline{\bm{H}})\|_{*}.
	\end{align}
	The second inequality uses the polarization identity:
	\begin{align}\label{56}
		\bm{AZC} - \bm{BZD} = (\bm{A-B})\bm{ZC} + \bm{BZ}(\bm{C-D})
	\end{align}
	and $ \|\bm{AB} \|_{*} \le \|\bm{A}\| \| \bm{B}\|_{*} $. Also, the last line used \eqref{37}, \eqref{38}, and the fact that $ \| \bm{L}_{22} \| \le \bm{L}$. First, define the following matrices:
	\begin{align}
		& \bm L ^ \prime := 
		\small{	\begin{bmatrix}
				\bm{0} & \bm{L}_{12} & &  \\
				& \bm{L}_{22}-\bm{I} & & \\
				& & \bm{\Lambda}_2 - \bm{I} & \\
				& & &\bm{0}
		\end{bmatrix}}\nonumber \\
		& \bm R^ \prime := 
		\small{\begin{bmatrix}
				\bm{0} & & & \\
				\bm{R}^{\rm H}_{12}& \bm{R}_{22}-\bm{I} & & \\
				& & \bm{\Gamma}_2 - \bm{I} & \\ 
				&  &  & \bm{0}
			\end{bmatrix}.}
	\end{align}
	\normalfont
	Now we upper bound the second term of \eqref{53}:
	\begin{align}
		&\| \bm{L}\overline{\bm{H}}^{\prime}\bm{R}^{\rm H}-\overline{\bm{H}}^{\prime} \|_{*} = \Bigg \| \bm{L} \overline{ \bm{H}}^{\prime}\bm{R}^{\rm H} -\begin{bmatrix}
			\bm{L}_{11} \!\!\!\!\!\!  &   \\
			&‌\bm{I}  \\
		\end{bmatrix}
		\overline{\bm{H}}^{\prime}\begin{bmatrix}
			\bm{R}_{11} \!\!\!  \!\!\! &  \\
			&‌\bm{I} 
		\end{bmatrix}\Bigg\|_{*} \nonumber \\
		& \stackrel{\eqref{56}}{\le} \| \bm L ^ \prime
		\overline{\bm{H}}^{\prime}
		\bm{R}^{\rm H} \|_* + \Bigg\|
		\begin{bmatrix}
			\bm{L}_{11} &  \\
			&‌\bm{I} 
		\end{bmatrix}
		\overline{\bm{H}}^{\prime} \bm R^ \prime
		\Bigg\|_{*} 
		\label{57}
		\le \| \bm L ^ \prime \| \|\overline{\bm{H}}^{ \prime} \|_{*} \|\bm{R}\|  \nonumber \\
		& + \max\{\|\bm{L}_{11}\| ,1\}
		\|\overline{\bm{H}}^{\prime} \|_{*} \| \bm R^ \prime
		\| \stackrel{\eqref{38}}{\le} ( \| \bm L ^ \prime \| + \| \bm R^ \prime \| )\| \overline{\bm{H}}^{\prime} \|_{*},
	\end{align}
	where we have used the fact that $  \|\bm{AB} \|_{*} \le \|\bm{A}\| \| \bm{B}\|_{*} $ and $\| \bm{L}_{11}  \| \le \| \bm{L}\| $. Replace \eqref{55} and \eqref{57} back into \eqref{53}: 
	\begin{align}\label{58}
		&‌ \|  \mathcal{ P}_{ \overline{ \bm{T}}^{ \perp}}( \overline{ \bm{H}}) \|_* \le ‌(\|[\bm{L}_{11} \quad  \bm{L}_{12}]\| +‌\|[\bm{R}_{11} \quad  \bm{R}_{12}]\|) \|\mathcal{P}_{\overline{\bm{T}}}(\overline{\bm{H}})\|_{*}   \nonumber \\
		& + (\| \bm L ^ \prime \| + \| \bm R^ \prime \| ) \|\overline{\bm{H}}^{\prime}\|_* + 2\| \bm{Q}_{\widetilde{\bm{\mathcal{U}}}_{r^{\prime}}}  \bm{X}_{r^{+}} \bm{Q}_{\widetilde{\bm{\mathcal{V}}}_{r^{\prime}}}  \|_{*} \stackrel{\eqref{eq:L11L12} , \eqref{eq:L'} , \eqref{eq:a3} , \eqref{eq:a4}}{=}  \nonumber \\
		& \alpha_3 \|\mathcal{P}_{ \overline{ \bm{T}}}( \overline{\bm{H}})\|_{*} + \alpha_4 \|\overline{\bm{H}}^{\prime} \|_{*} + 2\| \bm{Q}_{\widetilde{\bm{\mathcal{U}}}_{r^{\prime}}}  \bm{X}_{r^{+}} \bm{Q}_{\widetilde{\bm{\mathcal{V}}}_{r^{\prime}}}  \|_*.
	\end{align}
	According to \eqref{48} and the rotational invariance of the nuclear norm, it holds that 
	\begin{align} 
		\label{eq:h bar}
		&\|\mathcal{P}_{\overline{\bm{T}}}(\overline{\bm{H}})\|_{*} = \|\mathcal{P}_{\bm{T}}(\bm{H})\|_{*}, ‌\| \mathcal{P}_{\overline{\bm{T}}^{\perp}}(\overline{\bm{H}}) \|_{*} = ‌\| \mathcal{P}_{\bm{T}^{\perp}}(\bm{H}) \|_{*}.
	\end{align}
	Also, we define linear subspace $\widetilde{\bm{T}} \subset \bm{T}^{\perp} $:
	\begin{align}\label{60}
		&\widetilde{\bm{T}} := \Big\{\bm{Z}\in\mathbb{R}^{ n\times n} : \bm{Z}=\bm{B}_L \small{\begin{bmatrix}
				\bm 0 & & \\
				&\overline{\bm{Z}}_{22} & \overline{\bm{Z}}_{23}&  \overline{\bm{Z}}_{24}\\ 
				&\overline{\bm{Z}}_{32} &  \overline{\bm{Z}}_{33}& \overline{\bm{Z}}_{34} \\
				& \overline{\bm{Z}}_{42} & \overline{\bm{Z}}_{43} & \bm 0
		\end{bmatrix}}
		\normalfont \bm{B}_R^{\rm H} \Big\}.
	\end{align}
	Rotational invariance of the nuclear norm  and definition of $\overline{\bm{H}}^{\prime} $ in \eqref{51} yields:  
	\begin{align}
		\label{eq:h''}
		\|\overline{\bm{H}}^{\prime} \|_{*} = \|\bm{B}_L\overline{\bm{H}}^{\prime} \bm{B}_R^{\rm H}\|‌_{*} = \|\mathcal{P}_{\widetilde{\bm{T}}}(\bm{H})\|_{*}.
	\end{align}
	Finally, we rewrite \eqref{58} using \eqref{eq:h bar} and \eqref{eq:h''} as
	\begin{align}\label{62}
		‌\| \mathcal{P}_{\bm{T}^{\perp}}(\bm{H}) \|_{*} \le  \alpha_3 \|\mathcal{P}_{\bm{T}}(\bm{H})\|_{*} + \alpha_4 \|\mathcal{P}_{\widetilde{\bm{T}}}(\bm{H})\|_{*}+‌2 \|\bm{X}_{r^{+}}\|_{*},
	\end{align}
	where we have used the fact that $ \|\bm{AB} \|_{*} \le \| \bm{A}\| \|\bm{B} \|_{*}$ besides \eqref{38}.

	\section{Proof of Lemma \ref{lem 4}}\label{proof lemma 4}
	We use the fact that the operator norm of a diagonal matrix is its largest element. Also, for $ \bm{X} \in \mathbb{R}^{n \times n }$ 	
	\begin{equation}
		\| \bm{X} \| =  \sqrt{ \lambda_{ \max}( \bm{X}^{\rm H}\bm{X})} = \sigma_{ \max}( \bm{X}),
	\end{equation}
	where $\lambda_{\max}(\cdot)$ is the largest eigenvalue and $\sigma_{ \max}( \cdot )$ the largest singular value of a matrix.
	\begin{align}
		& \| \bm{L}_{11} \| = \| \bm{\Delta}_L\| = \max_{i} \sqrt{ \lambda_{1}^{2}(i)\cos^2 \mathbf{\theta}_{u}(i) + \sin^2 \mathbf{\theta}_{u}(i)}  \nonumber \\
		& \|\bm{L}_{12}\| = \max_{i} \sqrt{ \frac{( 1 - \lambda_{1}^{2} (i))^2 \cos^2 \mathbf{\theta}_{u}(i) + \sin^2 \mathbf{\theta}_{u}(i)}{\lambda_{1}^{2}(i)\cos^2 \mathbf{\theta}_{u}(i) + \sin^2 \mathbf{\theta}_{u}(i)} }   \nonumber \\
		& \|\bm I - \bm{L}_{22} \| = \| \bm I - \bm{\Lambda}\bm{\Delta}_L^{-1}\| = \nonumber \\
		& \max_i \frac{  \lambda_{1}(i) - \sqrt{ \lambda_i^2 \cos^2 \mathbf{\theta}_{u}(i) + \sin^2 \mathbf{\theta}_{u}(i)}} { \sqrt{ \lambda_{1}^{2} (i) \cos^2 \mathbf{ \theta}_{u}(i) + \sin^2 \mathbf{\theta}_{u}(i)}}  \nonumber \\
		& \Big\|
		\begin{bmatrix}
			\bm I - \bm{L}_{22} & \\
			& \bm I - \bm{\Lambda}_{2}
		\end{bmatrix} 
		\Big\| =  \max \{  \max_i (1- \lambda_{2}(i))\nonumber \\
		&‌ \quad \quad \quad \quad \quad \quad , \max_i \Big(1- \frac{\lambda_{1}(i)}{\sqrt{\lambda_{1}^{2}(i)\cos^2 \mathbf{\theta}_{u}(i) + \sin^2 \mathbf{\theta}_{u}(i)}} \Big) \}, \nonumber \\
		&\|[\bm{L}_{11} \quad  \bm{L}_{12}]\|^2 = \max_i \Bigg\|\Bigg[
		\sqrt{ \lambda_{1}^{2} (i) \cos^2 \mathbf{\theta}_{u}(i) + \sin^2 \mathbf{\theta}_{u}(i)} \nonumber \\
		&\quad \quad  \quad \quad \quad \quad  \quad \quad \frac{(1-\lambda_{1}^{2}(i)) \cos \mathbf{\theta}_{u}(i) \sin \mathbf{\theta}_{u}(i)}{\sqrt{ \lambda_{1}^{2}(i)\cos^2 \mathbf{\theta}_{u}(i) + \sin^2 \mathbf{\theta}_{u}(i)}}
		\Bigg]\Bigg\|^2_{2} \nonumber \\
		&= \max_i \frac{\lambda_{1}^{4}(i) \cos^2 \mathbf{\theta}_{u}(i) + \sin^2 \mathbf{\theta}_{u}(i)}{\lambda_{1}^{2}(i) \cos^2 \mathbf{\theta}_{u}(i) + \sin^2 \mathbf{\theta}_{u}(i) } \nonumber \\
		& \| \bm L^{\prime} \|^{2} = \max \Big\{ \Big\| \begin{bmatrix}
			\bm{L}_{12} \\  \bm{L}_{22}-\bm I
		\end{bmatrix}
		\Big\|_2^2  , \| \bm{\Lambda}_2 - \bm I \|_2^2 \Big\}	
		\nonumber \\
		& =\max_i \{[ \max_i \Big(1- \frac{\lambda_{1}^2 (i)} {\lambda_{1}^{2}(i)\cos^2 u_i + \sin^2 u_i} + \nonumber \\
		&\quad \quad \frac{(1-\lambda_{1}(i))^2\cos^2 \mathbf{\theta}_{u}(i)\sin^2 \mathbf{\theta}_{u}(i)}{\lambda_{1}^{2}(i)\cos^2 \mathbf{\theta}_{u}(i) + \sin^2 \mathbf{\theta}_{u}(i)} \Big) , \max_i (\lambda_{2}(i)-1)^2] \}.
	\end{align}
	\bibliographystyle{ieeetr}
\bibliography{MyrefrenceMRadMC}

\begin{thebibliography}{10}

\bibitem{FDD}
W.~Shen, L.~Dai, B.~Shim, S.~Mumtaz, and Z.~Wang, ``Joint csit acquisition
  based on low-rank matrix completion for fdd massive mimo systems,'' {\em IEEE
  Communications Letters}, vol.~19, no.~12, pp.~2178--2181, 2015.

\bibitem{haldar2010spatiotemporal}
J.~P. Haldar and Z.-P. Liang, ``Spatiotemporal imaging with partially separable
  functions: A matrix recovery approach,'' in {\em Biomedical Imaging: From
  Nano to Macro, 2010 IEEE International Symposium on}, pp.~716--719, 2010.

\bibitem{zhao2010low}
B.~Zhao, J.~P. Haldar, C.~Brinegar, and Z.-P. Liang, ``Low rank matrix recovery
  for real-time cardiac mri,'' in {\em Biomedical Imaging: From Nano to Macro,
  2010 IEEE International Symposium on}, pp.~996--999, 2010.

\bibitem{gross2010quantum}
D.~Gross, Y.-K. Liu, S.~T. Flammia, S.~Becker, and J.~Eisert, ``Quantum state
  tomography via compressed sensing,'' {\em Physical review letters}, vol.~105,
  no.~15, p.~150401, 2010.

\bibitem{srebro2010collaborative}
N.~Srebro and R.~R. Salakhutdinov, ``Collaborative filtering in a non-uniform
  world: Learning with the weighted trace norm,'' in {\em Advances in Neural
  Information Processing Systems}, pp.~2056--2064, 2010.

\bibitem{bennett2007netflix}
J.~Bennett, S.~Lanning, {\em et~al.}, ``The netflix prize,'' in {\em
  Proceedings of KDD cup and workshop}, vol.~2007, p.~35, New York, NY, USA,
  2007.

\bibitem{aravkin2014fast}
A.~Aravkin, R.~Kumar, H.~Mansour, B.~Recht, and F.~J. Herrmann, ``Fast methods
  for denoising matrix completion formulations, with applications to robust
  seismic data interpolation,'' {\em SIAM Journal on Scientific Computing},
  vol.~36, no.~5, pp.~S237--S266, 2014.

\bibitem{recht2010guaranteed}
B.~Recht, M.~Fazel, and P.~A. Parrilo, ``Guaranteed minimum-rank solutions of
  linear matrix equations via nuclear norm minimization,'' {\em SIAM review},
  vol.~52, no.~3, pp.~471--501, 2010.

\bibitem{so2007theory}
A.~M.-C. So and Y.~Ye, ``Theory of semidefinite programming for sensor network
  localization,'' {\em Mathematical Programming}, vol.~109, no.~2-3,
  pp.~367--384, 2007.

\bibitem{netflix}
A.~SIGKDD, ``Netflix,'' in {\em Proceedings of kdd cup and workshop}, 2007.

\bibitem{shen2016compressed}
J.-C. Shen, J.~Zhang, E.~Alsusa, and K.~B. Letaief, ``Compressed csi
  acquisition in fdd massive mimo: How much training is needed?,'' {\em IEEE
  Transactions on Wireless Communications}, vol.~15, no.~6, pp.~4145--4156,
  2016.

\bibitem{chen2020massive}
X.~Chen, D.~W.~K. Ng, W.~Yu, E.~G. Larsson, N.~Al-Dhahir, and R.~Schober,
  ``Massive access for 5g and beyond,'' {\em IEEE Journal on Selected Areas in
  Communications}, vol.~39, no.~3, pp.~615--637, 2020.

\bibitem{gao2015structured}
Z.~Gao, L.~Dai, W.~Dai, B.~Shim, and Z.~Wang, ``Structured compressive
  sensing-based spatio-temporal joint channel estimation for fdd massive
  mimo,'' {\em IEEE Transactions on Communications}, vol.~64, no.~2,
  pp.~601--617, 2015.

\bibitem{liang2019semi}
S.~Liang, X.~Wang, and L.~Ping, ``Semi-blind detection in hybrid massive mimo
  systems via low-rank matrix completion,'' {\em IEEE Transactions on Wireless
  Communications}, vol.~18, no.~11, pp.~5242--5254, 2019.

\bibitem{eftekhari2018weighted}
A.~Eftekhari, D.~Yang, and M.~B. Wakin, ``Weighted matrix completion and
  recovery with prior subspace information,'' {\em IEEE Transactions on
  Information Theory}, vol.~64, no.~6, pp.~4044--4071, 2018.

\bibitem{rao2015collaborative}
N.~Rao, H.-F. Yu, P.~K. Ravikumar, and I.~S. Dhillon, ``Collaborative filtering
  with graph information: Consistency and scalable methods,'' in {\em Advances
  in neural information processing systems}, pp.~2107--2115, 2015.

\bibitem{angst2011generalized}
R.~Angst, C.~Zach, and M.~Pollefeys, ``The generalized trace-norm and its
  application to structure-from-motion problems,'' in {\em 2011 International
  Conference on Computer Vision}, pp.~2502--2509, IEEE, 2011.

\bibitem{jain2013provable}
P.~Jain and I.~S. Dhillon, ``Provable inductive matrix completion,'' {\em arXiv
  preprint arXiv:1306.0626}, 2013.

\bibitem{xu2013speedup}
M.~Xu, R.~Jin, and Z.-H. Zhou, ``Speedup matrix completion with side
  information: Application to multi-label learning,'' in {\em Advances in
  neural information processing systems}, pp.~2301--2309, 2013.

\bibitem{mohan2010reweighted}
K.~Mohan and M.~Fazel, ``Reweighted nuclear norm minimization with application
  to system identification,'' in {\em Proceedings of the 2010 American Control
  Conference}, pp.~2953--2959, IEEE, 2010.

\bibitem{zhou2012kernelized}
T.~Zhou, H.~Shan, A.~Banerjee, and G.~Sapiro, ``Kernelized probabilistic matrix
  factorization: Exploiting graphs and side information,'' in {\em Proceedings
  of the 2012 SIAM international Conference on Data mining}, pp.~403--414,
  SIAM, 2012.

\bibitem{ardakani2019greedy}
H.~Ardakani, S.~Fazael, S.~Daei, and F.~Haddadi, ``A greedy algorithm for
  matrix recovery with subspace prior information,'' {\em arXiv preprint
  arXiv:1907.11868}, 2019.

\bibitem{daei2018optimal}
S.~Daei, A.~Amini, and F.~Haddadi, ``Optimal weighted low-rank matrix recovery
  with subspace prior information,'' {\em arXiv preprint arXiv:1809.10356},
  2018.

\bibitem{donoho2006compressed}
D.~L. Donoho, ``Compressed sensing,'' {\em IEEE Transactions on information
  theory}, vol.~52, no.~4, pp.~1289--1306, 2006.

\bibitem{candes2008restricted}
E.~J. Candes {\em et~al.}, ``The restricted isometry property and its
  implications for compressed sensing,'' {\em Comptes rendus mathematique},
  vol.~346, no.~9-10, pp.~589--592, 2008.

\bibitem{daei2019error}
S.~Daei, F.~Haddadi, A.~Amini, and M.~Lotz, ``On the error in phase transition
  computations for compressed sensing,'' {\em IEEE Transactions on Information
  Theory}, vol.~65, no.~10, pp.~6620--6632, 2019.

\bibitem{daei2019living}
S.~Daei, F.~Haddadi, and A.~Amini, ``Living near the edge: A lower-bound on the
  phase transition of total variation minimization,'' {\em IEEE Transactions on
  Information Theory}, vol.~66, no.~5, pp.~3261--3267, 2019.

\bibitem{needell2017weighted}
D.~Needell, R.~Saab, and T.~Woolf, ``Weighted-minimization for sparse recovery
  under arbitrary prior information,'' {\em Information and Inference: A
  Journal of the IMA}, vol.~6, no.~3, pp.~284--309, 2017.

\bibitem{lu2019compressive}
W.~Lu, Y.~Wang, X.~Wen, X.~Hua, S.~Peng, and L.~Zhong, ``Compressive downlink
  channel estimation for fdd massive mimo using weighted $ l\_ $\{$p$\}$ $
  minimization,'' {\em IEEE Access}, vol.~7, pp.~86964--86978, 2019.

\bibitem{love2008overview}
D.~J. Love, R.~W. Heath, V.~K. Lau, D.~Gesbert, B.~D. Rao, and M.~Andrews, ``An
  overview of limited feedback in wireless communication systems,'' {\em IEEE
  Journal on selected areas in Communications}, vol.~26, no.~8, pp.~1341--1365,
  2008.

\bibitem{liu2020angular}
G.~Liu, A.~Liu, R.~Zhang, and M.~Zhao, ``Angular-domain selective channel
  tracking and doppler compensation for high-mobility mmwave massive mimo,''
  {\em IEEE Transactions on Wireless Communications}, vol.~20, no.~5,
  pp.~2902--2916, 2020.

\bibitem{li2019time}
M.~Li, S.~Zhang, N.~Zhao, W.~Zhang, and X.~Wang, ``Time-varying massive mimo
  channel estimation: Capturing, reconstruction, and restoration,'' {\em IEEE
  Transactions on Communications}, vol.~67, no.~11, pp.~7558--7572, 2019.

\bibitem{qin2018sparse}
Q.~Qin, L.~Gui, B.~Gong, and S.~Luo, ``Sparse channel estimation for massive
  mimo-ofdm systems over time-varying channels,'' {\em IEEE Access}, vol.~6,
  pp.~33740--33751, 2018.

\bibitem{ma2018sparse}
J.~Ma, S.~Zhang, H.~Li, F.~Gao, and S.~Jin, ``Sparse bayesian learning for the
  time-varying massive mimo channels: Acquisition and tracking,'' {\em IEEE
  Transactions on Communications}, vol.~67, no.~3, pp.~1925--1938, 2018.

\end{thebibliography}
\vskip -2\baselineskip plus -1fil
\begin{IEEEbiography}[{\includegraphics[width=1in,height=1.25in,trim={0cm 0cm 0cm 0cm}]{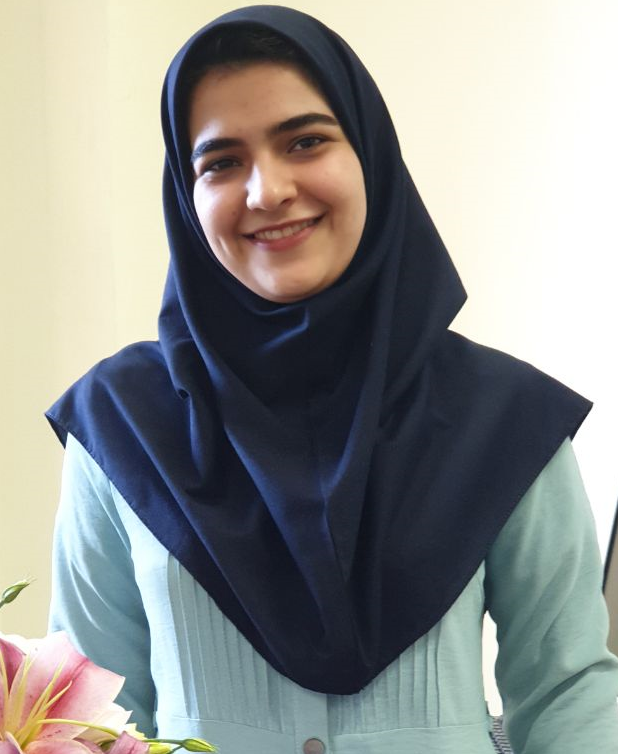}}]{\textbf{Hamideh Sadat Fazael Ardakani}}
received her B.Sc. degree in electronic engineering from Yazd university in 2017, Yazd, Iran and M.Sc. degree in communications engineering from Iran University of Science \& Technology, Tehran, Iran in 2020. She is currently pursuing her Ph.D. in the University of Tehran, Tehran, Iran. Her main research interests are matrix completion, compressed sensing and statistical signal processing.
\end{IEEEbiography}
\vskip -2\baselineskip plus -1fil
\begin{IEEEbiography}
	[{\includegraphics[width=1in,height=1.5in,clip,keepaspectratio]{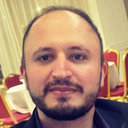}}]{\textbf{Sajad Daei}}
	received the B.Sc., degree in electronic engineering from Guilan University, Rasht, Iran, in 2011, the M.Sc. degree in communications engineering from Sharif University of Technology (SUT), Tehran, Iran, in 2013 and the Ph.D. degree in communications engineering from Iran University of Science \& Technology (IUST), Tehran, Iran in 2019. From 2020 to 2021, he was a research assistant in the electronic research institution of SUT. In 2020, he received the best Ph.D. thesis award of communication engineering and the outstanding Ph.D. thesis award of IEEE (Iran Section). He is currently a Postdoctoral researcher with EURECOM, Biot, France. His main research interests include optimization, inverse problems, compressed sensing and super resolution.
\end{IEEEbiography}
\vskip -2\baselineskip plus -1fil
\begin{IEEEbiography}[{\includegraphics[width=1in,height=1.25in,clip,keepaspectratio]{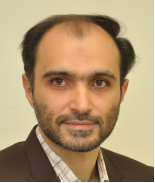}}]{\textbf{Farzan Haddadi}}was born in 1979. He received his B.Sc., M.Sc., and Ph.D. degrees in communication systems in 2001, 2003, and 2010, respectively, from Sharif University of Technology, Tehran, Iran. He joined Iran University of Science \& Technology faculty in 2011. His main research interests are array signal processing, statistical signal processing, subspace tracking, and compressed sensing.
\end{IEEEbiography}


\end{document}